\documentclass[12pt, letterpaper, onecolumn]{IEEEtran}

\usepackage{graphicx}
\usepackage{subfigure}
\usepackage{amsmath}
\usepackage{amssymb}
\usepackage{cite}
\usepackage{cuted}
\usepackage{stfloats}
\usepackage{color}
\usepackage{flushend}
\usepackage{epstopdf}
\IEEEoverridecommandlockouts
\newtheorem{theorem}{\bf Theorem}
\newtheorem{lemma}{\bf Lemma}

\newtheorem{definition}{\bf Definition}

\newcommand{\indi}[1]{{1\hspace{-2.3mm}{1}}_{\left\{#1\right\}}}
\newcommand{\m}[1]{\mathbf{#1}^m}
\newcommand{\lo}[1]{\log_2\left(#1\right)}
\newcommand{\lon}[1]{\ln\left(#1\right)}
\newcommand{\mk}[2]{\mathbf{#1}^{#2}}

\newcommand{\whatmn}[1]{\mathbf{\widehat{#1}}^m}

\newcommand{\co}[1]{\ensuremath{\underline{\underline{\text{C}_{#1}}}}}

\newcommand{\expect}[1]{\mathbb{E}\left[{#1}\right]}
\newcommand{\expectp}[2]{\mathbb{E}_{#1}\left[{#2}\right]}

\definecolor{DarkGreen2}{rgb}{0.00,0.6,0.08}

\title{The finite-dimensional {Witsenhausen} counterexample}

\author{Pulkit Grover,  Se Yong Park and Anant Sahai\\ Department of EECS, University of California at Berkeley, CA-94720, USA\\ \{pulkit, sahai, separk\}@eecs.berkeley.edu}
\begin{document}\maketitle
\vspace{-1in}
\begin{abstract}
Recently, a vector version of Witsenhausen's counterexample was considered and it was shown that in that limit of infinite vector length, certain quantization-based control strategies are provably within a constant factor of the optimal cost for all possible problem parameters. In this paper, finite vector lengths are considered with the dimension being viewed as an additional problem parameter. By applying a large-deviation ``sphere-packing'' philosophy, a lower bound to the optimal cost for the finite dimensional case is derived that uses appropriate shadows of the infinite-length bound.  Using the new lower bound, we show that good lattice-based control strategies achieve within a constant factor of the optimal cost uniformly over all possible problem parameters, including the vector length. For Witsenhausen's original problem ---  the scalar case --- the gap between regular lattice-based strategies and the lower bound is  numerically never more than a factor of $8$. 
\end{abstract}
\vspace{-0.3in}
\section{Introduction}
 Distributed control problems have long proved challenging
for control engineers. In 1968, Witsenhausen~\cite{Witsenhausen68}
gave a counterexample showing that even a seemingly simple distributed
control problem can be hard to solve. For the counterexample,
Witsenhausen chose a two-stage distributed LQG system and provided a
nonlinear control strategy that outperforms all linear laws. It is now
clear that the non-classical information pattern of Witsenhausen's
problem makes it quite challenging\footnote{In words of Yu-Chi
  Ho~\cite{YuChiHoCDC08}, ``the simplest problem becomes the hardest
  problem.''}; the optimal strategy and the optimal costs for the
problem are still unknown --- non-convexity makes
the search for an optimal strategy hard~\cite{bansalbasar,baglietto,
  LeeLauHo}. Discrete approximations of the problem~\cite{hochang} are
even NP-complete\footnote{More precisely, results in~\cite{papadimitriou} imply that the discrete counterparts to the Witsenhausen counterexample are NP-complete if the assumption of Gaussianity of the primitive random variables is relaxed. Further, it is also shown in~\cite{papadimitriou} that with this relaxation, a polynomial time solution to the original \textit{continuous} problem would imply $P=NP$, and thus conceptually the relaxed continuous problem is also hard.}~\cite{papadimitriou}.

In the absence of a solution, research on the counterexample has bifurcated into two different directions. Since there is no known systematic approach to obtain provably optimal solutions, a body of literature (e.g.~\cite{baglietto}~\cite{LeeLauHo}~\cite{marden} and the references therein) applies search heuristics to explore the space of possible control actions and obtain intuition into the structure of good strategies. Work in this direction has also yielded considerable insight into addressing non-convex problems in general.

In the other direction the emphasis is on understanding the role of \textit{implicit communication} in the counterexample. In distributed control, control actions not only attempt to reduce the immediate control costs, they can also communicate relevant information to other controllers to help them reduce costs. Witsenhausen~\cite[Section 6]{Witsenhausen68} and Mitter and Sahai~\cite{AreaExamPaper} aim at developing systematic constructions based on implicit communication. Witsenhausen's two-point quantization strategy is motivated from the optimal strategy for two-point symmetric distributions of the initial state~\cite[Section 5]{Witsenhausen68} and it outperforms linear strategies for certain parameter choices. Mitter and Sahai~\cite{AreaExamPaper} propose multipoint-quantization strategies that, depending on the problem parameters, can outperform linear strategies by an arbitrarily-large factor. 

Various modifications to the counterexample investigate if  misalignment of these two goals of control and implicit communication makes the problems hard~\cite{bansalbasar,basarCDC08,rotkowitzCDC08,RotkowitzLall,Allerton09Paper,rotkowitz} (see~\cite{WitsenhausenJournal} for a survey of other such modifications). Of particular interest are two  works, those of Rotkowitz and Lall~\cite{RotkowitzLall}, and Rotkowitz~\cite{rotkowitz}. The first work~\cite{RotkowitzLall} shows that with extremely fast, infinite-capacity, and perfectly reliable external channels, the optimal controllers are linear not just for the Witsenhausen's counterexample (which is a simple observation), but for more general problems as well. This suggests that allowing for an external channel between the two controllers in Witsenhausen's counterexample might simplify the problem. However, when the channel is not perfect, Martins~\cite{MartinsSideInfo} shows that finding optimal solutions can be hard\footnote{Martins shows that nonlinear strategies that do not even use the external channel can outperform linear ones that do use the channel where the external channel SNR is high. As is suggested by what David Tse calls the ``deterministic perspective" (along the lines of~\cite{DeterministicModel,SalmanThesis,DeterministicApproach}),  linear strategies do not make good use of the external channel because they only communicate the ``most significant bits'' --- which can anyway be estimated reliably at the second controller. So if the uncertainty in the initial state is large, the external channel is only of limited help and there may be substantial advantage in having the controllers talk through the plant.  A similar problem is considered by Shoarinejad et al in~\cite{shoarinejad}, where noisy side information of the source is available at the receiver. Since this formulation is even more constrained than that in~\cite{MartinsSideInfo}, it is clear that nonlinear strategies outperform linear for this problem as well.}. A closer inspection of the problem in~\cite{MartinsSideInfo} reveals that nonlinear strategies can outperform linear ones by an arbitrarily large factor for any fixed SNR on the external channel. Even to make good use of the external channel resource, one needs nonlinear strategies.

The second work~\cite{rotkowitz} shows that if one considers the induced norm instead of the original expected quadratic cost, linear control laws are optimal and easy to find. The induced norm formulation is therefore easy to solve, and at the same time, it makes no assumptions on the state and the noise distributions. This led Doyle to ask if Witsenhausen's counterexample (with expected quadratic cost) is at all relevant~\cite{PathsAhead} --- after all, not only is the LQG formulation more constrained, it is also harder to solve.
The question thus becomes what norm is more appropriate, and the answer must come from what is relevant in practical situations. In practice, one usually knows the ``typical" amplitude of the noise and the initial state, or at least  rough bounds them. The induced-norm formulation may therefore be quite conservative: since no assumptions are made on the state and the noise, it requires budgeting for completely arbitrary behavior of state and noise --- they can even collude to raise the costs for the chosen strategy. To see how conservative the induced-norm formulation can be, notice the following: even allowing for colluding state and noise, mere knowledge of a bound on the noise amplitude suffices to have quantization-based nonlinear strategies outperform linear strategies by an arbitrarily large factor (with the expected cost replaced by a hard-budget. The proof is simpler than that in~\cite{AreaExamPaper}, and is left as an exercise to the interested reader for reasons of limited space). Conceptually, the LQG formulation is only abstracting some knowledge of noise and initial state behavior. In practical situations where such knowledge exists, designs based on an induced norm formulation (and linear strategies) may be needlessly expensive because they budget for impossible events.



The fact that nonlinear strategies can be arbitrarily better brings us to a question that has received little attention in the literature --- how far are the proposed nonlinear strategies from the optimal? It is believed that the strategies of Lee, Lau and Ho~\cite{LeeLauHo} are close to optimal. In Section~\ref{sec:conclusions}, we will see that these strategies can be viewed as an instance of the ``dirty-paper coding'' strategy in information theory, and quantify their advantage over pure quantization based strategies. Despite their improved performance, there was no guarantee that these strategies are indeed close to optimal\footnote{The search in~\cite{LeeLauHo} is not exhaustive. The authors first find a good quantization-based solution. Inspired by piecewise linear strategies (from the neural networks based search of Baglietto \textit{et al}~\cite{baglietto}), each quantization step is broken into several small sub-steps to approximate a piecewise linear curve. 
}. Witsenhausen~\cite[Section 7]{Witsenhausen68} derived a lower bound on the costs that is loose in the interesting regimes of small $k$ and large $\sigma_0^2$~\cite{WitsenhausenJournal,CDC09paper}, and hence is insufficient to obtain any guarantee on the gap from optimality.


Towards obtaining such a guarantee, a strategic simplification of the problem was introduced in~\cite{CDCWitsenhausen,WitsenhausenJournal} where we consider an asymptotically-long vector version of the problem. This problem is related to a toy communication problem that we call ``Assisted Interference Suppression'' (AIS) which is an extension of the dirty-paper coding (DPC)~\cite{CostaDirtyPaper} model in information theory. There has been a burst of interest in extensions to DPC in information theory mainly along two lines of work --- multi-antenna Gaussian channels, and the ``cognitive-radio channel.'' For multi-antenna Gaussian channels, a problem of much theoretical and practical interest, DPC turns out to be the optimal strategy (see~\cite{ShamaiBroadcast} and the references therein). The ``cognitive radio channel" problem was formulated by Devroye \textit{et al}~\cite{Devroye1}. This inspired much work in asymmetric cooperation between nodes~\cite{JovicicViswanath,KimStateAmplification,MerhavMasking,KhistiLatticeMDPC,KotagiriLaneman}. In our work~\cite{WitsenhausenJournal,CDCWitsenhausen}, we developed a new lower bound to the optimal performance of the vector Witsenhausen problem. Using this bound, we show that  vector-quantization based strategies attain within a factor of $4.45$ of the optimal cost for all problem parameters in the limit of infinite vector length. Further, combinations of linear and DPC-based strategies attain within a factor $2$ of the optimal cost. This factor was later improved to $1.3$ in~\cite{ITW09Paper} by improving the lower bound. While a constant-factor result does not establish true optimality, such results are often helpful in the face of intractable problems like those that are otherwise NP-hard \cite{ApproximationBook}. This constant-factor spirit has also been useful in understanding other stochastic control problems
\cite{CogillLall06, CogillLallHespanha07} and in the asymptotic analysis of  problems in multiuser wireless communication \cite{EtkinOneBit, DeterministicModel}.

While the lower bound in~\cite{WitsenhausenJournal} holds for all vector lengths, and hence for the scalar counterexample as well, the ratio of the costs attained by the strategies of~\cite{AreaExamPaper} and the lower bound diverges in the limit
$k\rightarrow 0$ and $\sigma_0\rightarrow\infty$. This suggests that there is a significant finite-dimensional aspect of the problem that
is being lost in the infinite-dimensional limit: either quantization-based strategies are bad, or the lower bound of~\cite{WitsenhausenJournal} is very loose. This effect is
elucidated in \cite{CDC09paper} by deriving a different lower bound showing 
that quantization-based strategies indeed attain within a
constant\footnote{The constant is large in \cite{CDC09paper}, but as
  this paper shows, this is an artifact of the proof rather than
  reality.} factor of the optimal cost for Witsenhausen's original
problem. The bound in~\cite{CDC09paper} is in the spirit of Witsenhausen's original lower
bound, but is more intricate. It captures the idea that observation
noise can force a second-stage cost to be incurred unless the first
stage cost is large.

In this paper, we revert to the line of attack initiated by the vector
simplification of \cite{WitsenhausenJournal}. In Section~\ref{sec:notation}, we formally state the vector version of the counterexample. For obtaining
good control strategies, we observe that the action of the first controller in the quantization-based strategy of~\cite{AreaExamPaper} can be thought of as forcing the state to a point on a one-dimensional
\textit{lattice}. Extending this idea, in Section~\ref{sec:lattice}, we provide lattice-based quantization strategies for finite dimensional spaces and analyze their performance.

Building upon the vector
lower bound of~\cite{WitsenhausenJournal}, a new lower bound is derived in Section~\ref{sec:lowerbound} which is in the spirit of large-deviations-based information-theoretic bounds for finite-length communication problems\footnote{An alternative Central Limit Theorem (CLT)-based approach has also been used in the information-theory literature~\cite{Baron,VerduCLT,VerduDispersion}. In~\cite{VerduCLT,VerduDispersion}, the approach is used to obtain extremely tight approximations at moderate blocklengths for Shannon's noisy communication problem.}
(e.g.~\cite{Gallager,PinskerNoFeedback,OurUpperBoundPaper,
  waterslide}). In particular, our new bound extends the tools
in~\cite{waterslide} to a setting with unbounded distortion measure. In Section~\ref{sec:ratio}, we combine the lattice-based upper bound (Section~\ref{sec:lattice}) and the large-deviations lower bound (Section~\ref{sec:lowerbound}) to show that lattice-based quantization strategies attain within a constant factor of the optimal cost for any finite length, uniformly over all problem parameters. For example, this constant factor is numerically found to be smaller than $8$ for the original scalar problem. We also provide a constant factor that holds uniformly over all vector lengths.

To understand the significance of the result, consider the following. At $k=0.01$ and $\sigma_0=500$, the cost attained by the optimal linear scheme is close to $1$. The cost attained by a quantization-based\footnote{The quantization points are regularly spaced about $9.92$ units apart. This results in a first stage cost of about $8.2\times 10^{-4}$ and a second stage cost of about $6.7\times 10^{-5}$.} scheme is $8.894\times 10^{-4}$. Our new lower bound on the cost is $3.170\times 10^{-4}$.  Despite the small value of the lower bound,
the ratio of the quantization-based upper bound and the lower bound
for this choice of parameters is less than three!

We conclude in  Section~\ref{sec:conclusions} outlining  directions of future research and  speculating on the form of finite-dimensional strategies (following~\cite{WitsenhausenJournal}) that we conjecture might be optimal.




\section{Notation and problem statement}
\label{sec:notation}
\begin{figure}[htb]
\begin{center}
   \includegraphics[scale=0.45]{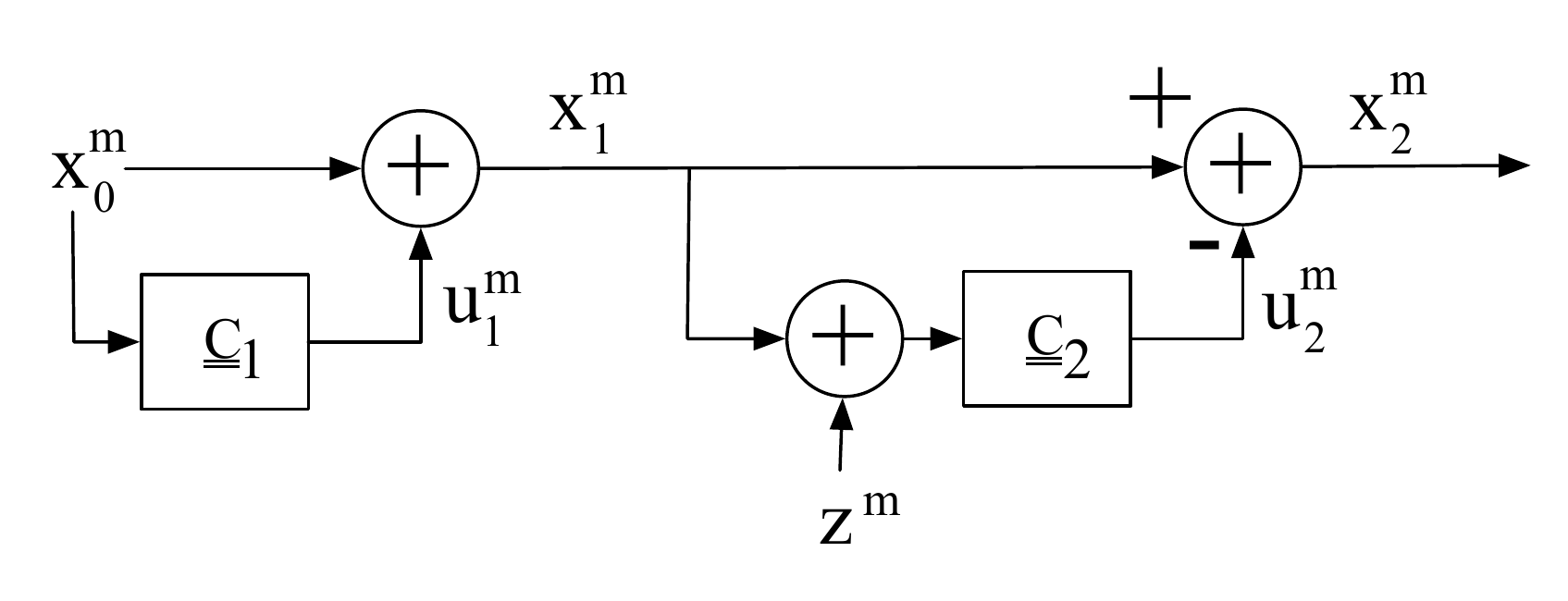}
\caption{Block-diagram for vector version of Witsenhausen's counterexample of length $m$. }
\label{fig:infotheory}
\end{center}
\end{figure}

Vectors are denoted in bold. Upper case tends to be used for random
variables, while lower case symbols represent their realizations. $W(m,k^2,\sigma_0^2)$ denotes the vector version of
Witsenhausen's problem of length $m$, defined as follows (shown in Fig.~\ref{fig:infotheory}):
\begin{itemize}
\item The initial state $\m{X}_0$ is Gaussian, distributed $\mathcal{N}(0,\sigma_0^2\mathbb{I}_m)$, where $\mathbb{I}_m$ is the identity matrix of size $m\times m$. 
\item The state transition functions describe the state evolution with time. The state transitions are linear: 
\begin{eqnarray*}
\m{X}_1 &=&\m{X}_0+\m{U}_1,\;\;\;\text {and}\\
\m{X}_2 &=&\m{X}_1-\m{U}_2.
\end{eqnarray*}
\item The outputs observed by the controllers: 
\begin{eqnarray}
\nonumber\m{Y}_1&=&\m{X}_0,\;\;\;\text{ and}\\
\m{Y}_2&=&\m{X}_1+\m{Z},
\label{eq:outputs}
\end{eqnarray}
where $\m{Z}\sim \mathcal{N}(0,\sigma_Z^2\mathbb{I}_m)$ is Gaussian distributed observation noise.  
\item The control objective is to minimize the expected cost, averaged over the  random realizations of $\m{X}_0$ and $\m{Z}$. The total cost is a quadratic function of the state and the input given by the sum of two terms:
\begin{eqnarray*}
J_1(\m{x}_1,\m{u}_1) &=& \frac{1}{m}k^2\|\m{u}_1\|^2,\; \text{and}\\
J_2(\m{x}_2,\m{u}_2)&=&\frac{1}{m}\|\m{x}_2\|^2
\end{eqnarray*}
where $\|\cdot\|$ denotes the usual Euclidean 2-norm.
The cost expressions are normalized by the vector-length $m$ to allow for natural comparisons between different vector-lengths. A control strategy is denoted by $\gamma=(\gamma_1,\gamma_2)$, where $\gamma_i$ is the function that maps the observation $\m{y}_i$ at $\co{i}$ to the control input $\m{u}_i$. For a fixed $\gamma$, $\m{x}_1=\m{x}_0+\gamma_1(\m{x}_0)$ is a function of $\m{x}_0$. Thus the first stage cost can instead be written as a function $J_1^{(\gamma)}(\m{x}_0)=J_1(\m{x}_0+\gamma_1(\m{x}_0),\gamma_1(\m{x}_0))$ and the second stage cost can be written as $J_2^{(\gamma)}(\m{x}_0,\m{z})=J_2(\m{x}_0+\gamma_1(\m{x}_0)-\gamma_2(\m{x}_0+\gamma_1(\m{x}_0)+\m{z}),\gamma_2(\m{x}_0+\gamma_1(\m{x}_0)+\m{z}))$.

For given $\gamma$, the expected costs (averaged over $\m{x}_0$ and $\m{z}$) are denoted by $\bar{J}^{(\gamma)}(m,k^2,\sigma_0^2)$ and $\bar{J}_i^{(\gamma)}(m,k^2,\sigma_0^2)$ for $i=1,2$. We define $\bar{J}^{(\gamma)}_{\min}(m,k^2,\sigma_0^2)$ as follows
\begin{equation}
\bar{J}_{\min}(m,k^2,\sigma_0^2):=\inf_{\gamma}\bar{J}^{(\gamma)}(m,k^2,\sigma_0^2).
\end{equation}

\end{itemize}
We note that for the scalar case of $m=1$, the problem is Witsenhausen's original counterexample~\cite{Witsenhausen68}.

 
Observe that scaling $\sigma_0$ and $\sigma_Z$ by the same factor essentially does not change the problem --- the solution can also be scaled by the same factor (with the resulting cost scaling quadratically with it). Thus, without loss of generality, we assume that the variance of the Gaussian observation noise is $\sigma_Z^2=1$  (as is also assumed in~\cite{Witsenhausen68}). The pdf of the noise $\m{Z}$ is denoted by $f_Z(\cdot{})$. In our proof techniques, we also consider a hypothetical observation noise $\m{Z}_G\sim\mathcal{N}(0,\sigma_G^2)$ with the variance $\sigma_G^2\geq 1$. The pdf of this test noise is denoted by $f_G(\cdot{})$. We use $\psi(m,r)$ to denote $\Pr(\|\m{Z}\|\geq r)$ for $\m{Z}\sim\mathcal{N}(0,\mathbb{I})$. 

Subscripts in expectation expressions denote the random variable being averaged over (e.g. $\expectp{\m{X}_0,\m{Z}_G}{\cdot{}}$ denotes averaging over the initial state $\m{X}_0$ and the test noise $\m{Z}_G$). 


\section{Lattice-based quantization strategies}
\label{sec:lattice}
\begin{figure}
\begin{center}
\includegraphics[width=9cm]{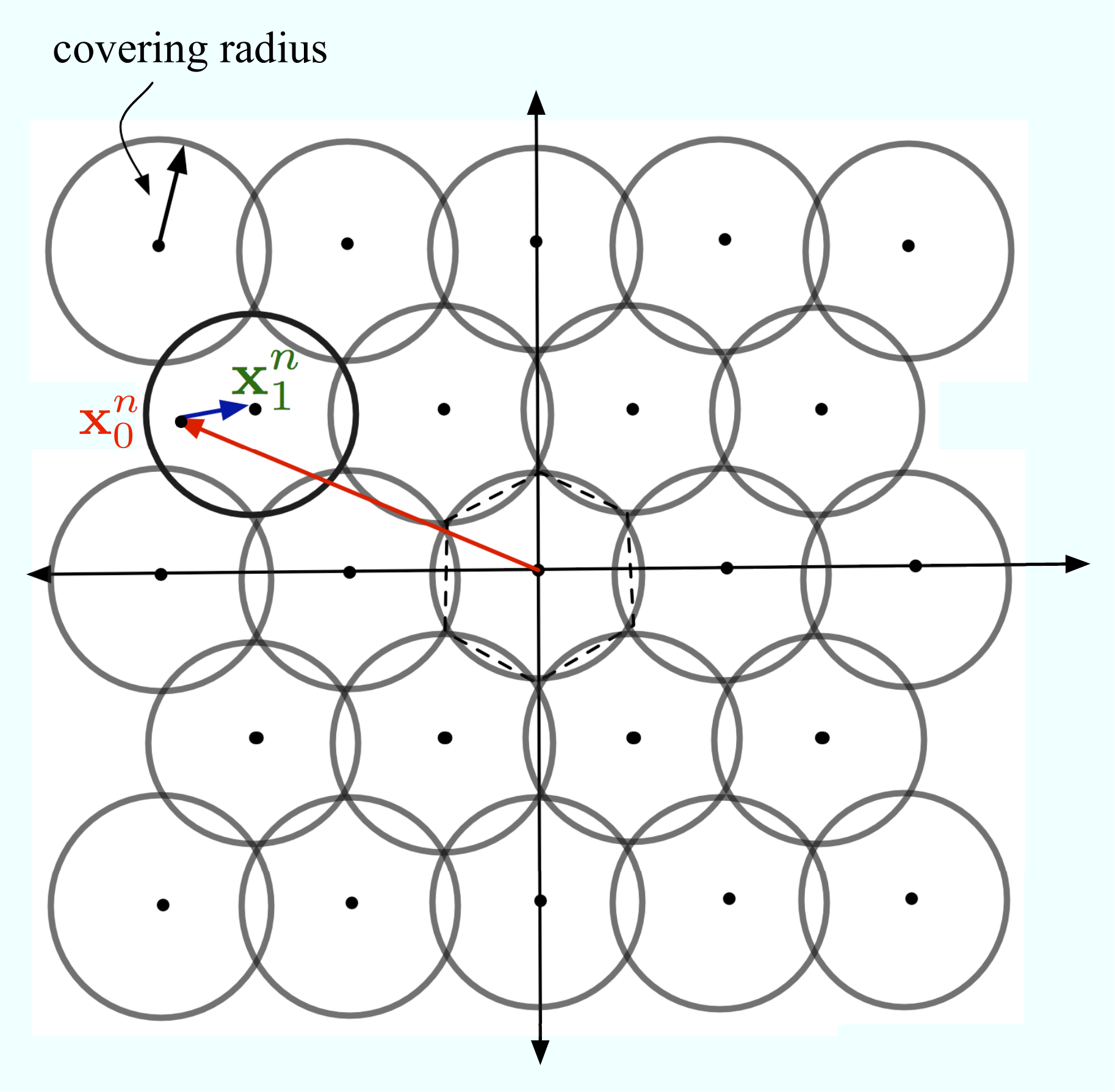}
\includegraphics[width=9cm]{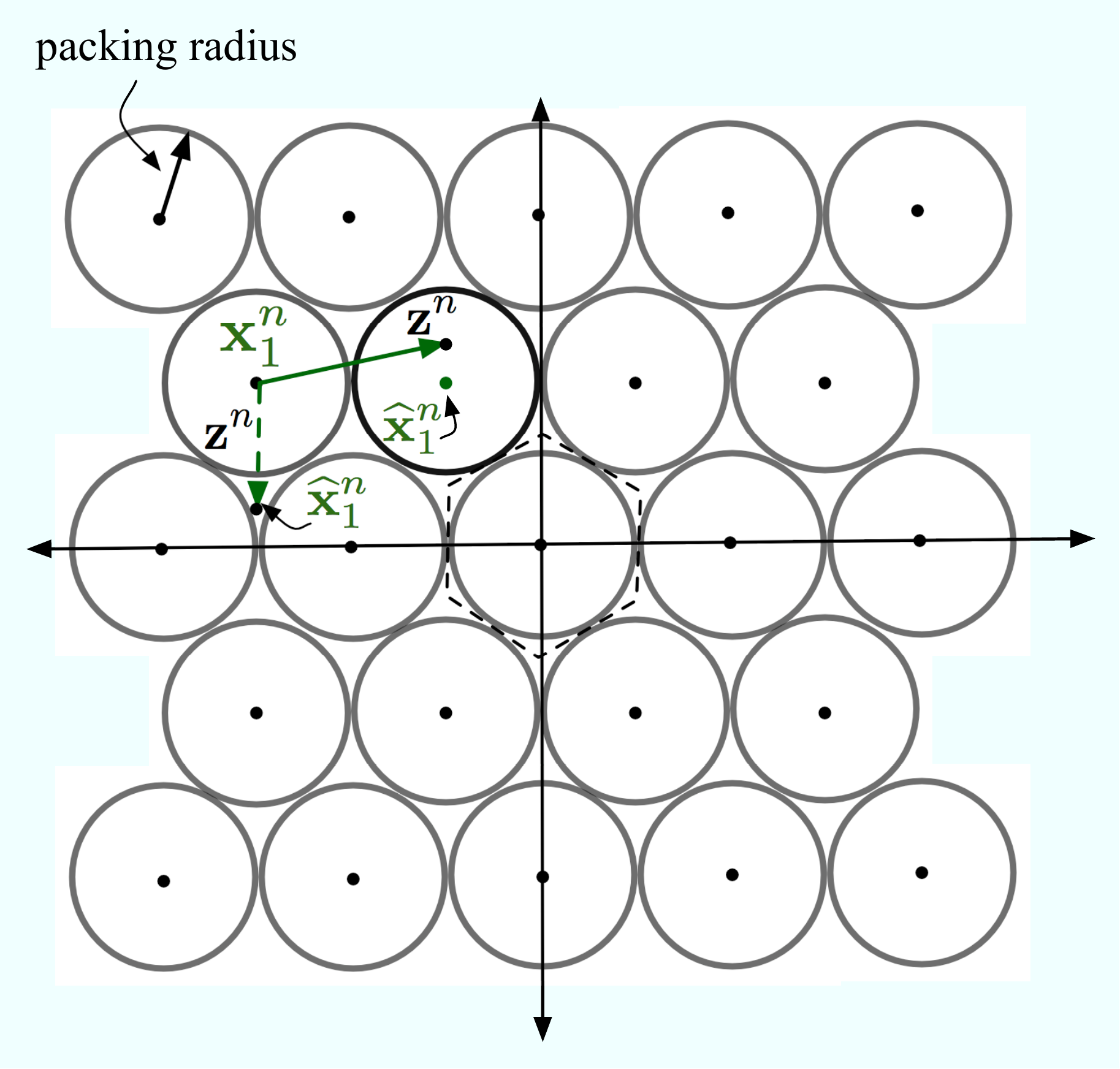}
\end{center}
\caption{Covering and packing for the 2-dimensional hexagonal
  lattice. The packing-covering ratio for this lattice is $\xi=\frac{2}{\sqrt{3}}
  \approx 1.15$~\cite[Appendix C]{Fischer}. The first controller
  forces the initial state $\m{x}_0$ to the lattice point nearest to
  it. The second controller estimates $\m{\widehat{x}}_1$ to be a
  lattice point at the centre of the sphere if it falls in one of the
  packing spheres. Else it essentially gives up and estimates
  $\m{\widehat{x}}_1=\m{y}_2$, the received output itself. A
  hexagonal lattice-based scheme would perform better for the 2-D
  Witsenhausen problem than the square lattice (of
  $\xi=\sqrt{2}\approx1.41$~\cite[Appendix C]{Fischer}) because it has
  a smaller $\xi$.}
\label{fig:lattice}
\end{figure}

Lattice-based quantization strategies are the natural generalizations of scalar quantization-based strategies~\cite{AreaExamPaper}. An introduction to lattices can be found
in~\cite{SloaneLattices,almosteverything}. Relevant definitions are
reviewed below. $\mathcal{B}$ denotes the unit ball in $\mathbb{R}^m$.
\begin{definition}[Lattice]
An $m$-dimensional lattice $\Lambda$ is a set of points in $\mathbb{R}^m$ such that if $\m{x},\m{y}\in\Lambda$, then $\m{x}+\m{y}\in \Lambda$, and if $\m{x}\in\Lambda$, then $-\m{x}\in\Lambda$.
\end{definition}

\begin{definition}[Packing and packing radius]
Given an $m$-dimensional lattice $\Lambda$ and a radius $r$, the set
$\Lambda+r\mathcal{B}$ is a \textit{packing} of Euclidean $m$-space if
for all points $\m{x},\m{y}\in\Lambda$,
$(\m{x}+r\mathcal{B})\bigcap(\m{y}+r\mathcal{B})=\emptyset$. The packing
radius $r_p$ is defined as $r_p:=\sup\{r:\Lambda+ r\mathcal{B}
\;\text{is a packing} \}$.  
\end{definition}

\begin{definition}[Covering and covering radius]
Given an $m$-dimensional lattice $\Lambda$ and a radius $r$, the set
$\Lambda+r\mathcal{B}$ is a \textit{covering} of Euclidean $m$-space
if $\mathbb{R}^m\subseteq \Lambda + r\mathcal{B}$. The covering radius
$r_c$ is defined as $r_c:=\inf\{r:\Lambda+ r\mathcal{B} \;\text{is a
  covering} \}$.  
\end{definition}

\begin{definition}[Packing-covering ratio]
The \textit{packing-covering ratio} (denoted by $\xi$) of a lattice
$\Lambda$ is the ratio of its covering radius to its packing radius,
$\xi=\frac{r_c}{r_p}$. 
\end{definition}

Because it creates no ambiguity, we do not include the dimension $m$
and the choice of lattice $\Lambda$ in the notation of $r_c$, $r_p$
and $\xi$, though these quantities depend on $m$ and
$\Lambda$.   


For a given dimension $m$, a natural control strategy that uses a lattice $\Lambda$ of covering radius $r_c$ and packing radius $r_p$ is as follows. The first controller uses the input $\m{u}_1$ to force the state $\m{x}_0$ to the lattice point nearest to $\m{x}_0$. The second controller estimates $\m{x}_1$ to be the lattice point nearest to $\m{y}_2$. For analytical ease, we instead consider an inferior strategy where the second controller estimates $\m{x}_1$ to be a lattice point only if the lattice point lies within the sphere of radius $r_p$ around $\m{y}_2$. If no lattice point exists in the sphere, the second controller estimates $\m{x}_1$ to be $\m{y}_2$, the received vector itself. The actions $\gamma_1(\cdot{})$ of $\co{1}$ and $\gamma_2(\cdot{})$ of $\co{2}$ are therefore given by
\begin{eqnarray*}
\gamma_1(\m{x}_0) &=& -\m{x}_0+\underset{{\m{x}_1}\in\Lambda }{\text{arg min}}\;\|\m{x}_1-\m{x}_0\|^2,\\
\gamma_2(\m{y}_2) &=& \left\{\begin{array}{cc}
\m{\widetilde{x}}_1 & \text{if}\;\exists\;  \m{\widetilde{x}}_1\in\Lambda\;\text{s.t.}\; \|\m{y}_2-\m{\widetilde{x}}_1\|^2< r_p^2\\
\m{y}_2 & \text{otherwise}
\end{array}\right. .
\end{eqnarray*}
The event where there exists no such $\m{\widetilde{x}}_1\in\Lambda$ is referred to as \textit{decoding failure}. In the following, we denote $\gamma_2(\m{y}_2)$ by $\m{\widehat{x}}_1$, the estimate of $\m{x}_1$. 

\begin{theorem}
\label{thm:upperbound}
Using a lattice-based strategy (as described above) for $W(m,k^2,\sigma_0^2)$  with $r_c$ and $r_p$ the covering and the packing radius for the lattice, the total average cost is upper bounded by
\begin{eqnarray*}
\bar{J}^{(\gamma)}(m,k^2,\sigma_0^2)\leq \inf_{P\geq 0} k^2P + 
\left(\sqrt{\psi(m+2,r_p)}+\sqrt{\frac{P}{\xi^2}}\sqrt{\psi(m,r_p)}\right)^2,
\end{eqnarray*}
where $\xi=\frac{r_c}{r_p}$ is the packing-covering ratio for the 
lattice, and $\psi(m,r)=\Pr(\|\m{Z}\|\geq r)$. The following looser bound also holds
\begin{eqnarray*}
\bar{J}^{(\gamma)}(m,k^2,\sigma_0^2)
\leq \inf_{P> \xi^2} 
k^2P+\left(1+\sqrt{\frac{P}{\xi^2}}\right)^2  
e^{-\frac{mP}{2\xi^2}+\frac{m+2}{2}\left(1+\lon{\frac{P}{\xi^2}}\right)}.
\end{eqnarray*}
\end{theorem}
\textit{Remark}: The latter loose bound is useful for analytical manipulations when proving explicit bounds on the ratio of the upper and lower bounds in Section~\ref{sec:ratio}.

\begin{proof}
Note that because $\Lambda$ has a covering radius of $r_c$, $\|\m{x}_1-\m{x}_0\|^2\leq r_c^2$. Thus the first stage cost is bounded above by $\frac{1}{m}k^2r_c^2$. A tighter bound can be provided for a specific lattice and finite $m$ (for example, for $m=1$, the first stage cost is approximately $k^2\frac{r_c^2}{3}$ if $r_c^2 \ll \sigma_0^2$ because the distribution of $\m{x}_0$ conditioned on it lying in any of the quantization bins is approximately uniform at least for the most likely bins). 

For the second stage, observe that
\begin{eqnarray}
\label{eq:eachpoint}
\expectp{\m{X}_1,\m{Z}}{\|\m{X}_1-\m{\widehat{X}}_1\|^2}=\expectp{\m{X}_1}{\expectp{\m{Z}}{\|\m{X}_1-\m{\widehat{X}}_1\|^2|\m{X}_1}}. 
\end{eqnarray}
Denote by $\mathcal{E}_m$ the event $\{\|\m{Z}\|^2\geq r_p^2\}$. Observe that under the event $\mathcal{E}_m^c$, $\m{\widehat{X}}_1=\m{X}_1$, resulting in a zero second-stage cost. Thus,
\begin{eqnarray*}
\expectp{\m{Z}}{\|\m{X}_1-\m{\widehat{X}}_1\|^2|\m{X}_1}&=&\expectp{\m{Z}}{\|\m{X}_1-\m{\widehat{X}}_1\|^2\indi{\mathcal{E}_m}|\m{X}_1}+\expectp{\m{Z}}{\|\m{X}_1-\m{\widehat{X}}_1\|^2\indi{\mathcal{E}_m^c}|\m{X}_1}\\
&=&\expectp{\m{Z}}{\|\m{X}_1-\m{\widehat{X}}_1\|^2\indi{\mathcal{E}_m}|\m{X}_1}.
\end{eqnarray*}
We now bound the squared-error under the error event $\mathcal{E}_m$, when either $\m{x}_1$ is decoded erroneously, or there is a decoding failure. If $\m{x}_1$ is decoded erroneously to a lattice point $\m{\widetilde{x}}_1\neq \m{x}_1$, the squared-error can be bounded as follows
\begin{eqnarray*}
\|\m{x}_1-\m{\widetilde{x}}_1\|^2  = \|\m{x}_1-\m{y}_2+\m{y}_2- \m{\widetilde{x}}_1\|^2
\leq  \left(\|\m{x}_1-\m{y}_2\|+\|\m{y}_2- \m{\widetilde{x}}_1\|\right)^2
\leq  \left(\|\m{z}\|+r_p\right)^2.
\end{eqnarray*}
If $\m{x}_1$ is decoded as $\m{y}_2$, the squared-error is simply
$\|\m{z}\|^2$, which we also upper bound by
$\left(\|\m{z}\|+r_p\right)^2$. Thus, under event $\mathcal{E}_m$,
the squared error $\|\m{x}_1-\m{\widehat{x}}_1\|^2$ is bounded above
by $\left(\|\m{z}\|+r_p\right)^2$, and hence
\begin{eqnarray}
\label{eq:andhence}
\expectp{\m{Z}}{\|\m{X}_1-\m{\widehat{X}}_1\|^2|\m{X}_1}&\leq& \expectp{\m{Z}}{\left(\|\m{Z}\|+r_p\right)^2\indi{\mathcal{E}_m}|\m{X}_1}\nonumber\\
&\overset{(a)}{=}&\expectp{\m{Z}}{\left(\|\m{Z}\|+r_p\right)^2\indi{\mathcal{E}_m}},
\end{eqnarray}
where $(a)$ uses the fact that the pair $(\m{Z},\indi{\mathcal{E}_m})$ is independent of $\m{X}_1$. Now, let $P=\frac{r_c^2}{m}$, so that the first stage cost is at most
$k^2P$. The following lemma helps us derive the upper bound. 
\begin{lemma}
\label{lem:upperbound}
For a given lattice with $r_p^2=\frac{r_c^2}{\xi^2}=\frac{mP}{\xi^2}$, the following bound holds 
\begin{eqnarray*}
\frac{1}{m}\expectp{\m{Z}}{\left(\|\m{Z}\|+r_p\right)^2\indi{\mathcal{E}_m}}
\leq \left(\sqrt{\psi(m+2,r_p)}+\sqrt{\frac{P}{\xi^2}}\sqrt{\psi(m,r_p)}\right)^2.
\end{eqnarray*}
The following (looser) bound also holds as long as $P>\xi^2$,
\begin{eqnarray*}
\frac{1}{m}\expectp{\m{Z}}{\left(\|\m{Z}\|+r_p\right)^2\indi{\mathcal{E}_m}}\leq \left(1+\sqrt{\frac{P}{\xi^2}}\right)^2e^{-\frac{mP}{2\xi^2}+\frac{m+2}{2}\left(1+\lon{\frac{P}{\xi^2}}\right)}.
\end{eqnarray*}
\end{lemma}
\begin{proof}
See Appendix~\ref{app:upperbound}.
\end{proof}
The theorem now follows from~\eqref{eq:eachpoint},~\eqref{eq:andhence} and Lemma~\ref{lem:upperbound}.
\end{proof}

\section{Lower bounds on the cost}
\label{sec:lowerbound}

\begin{figure}
\begin{center}
\includegraphics[width=9cm]{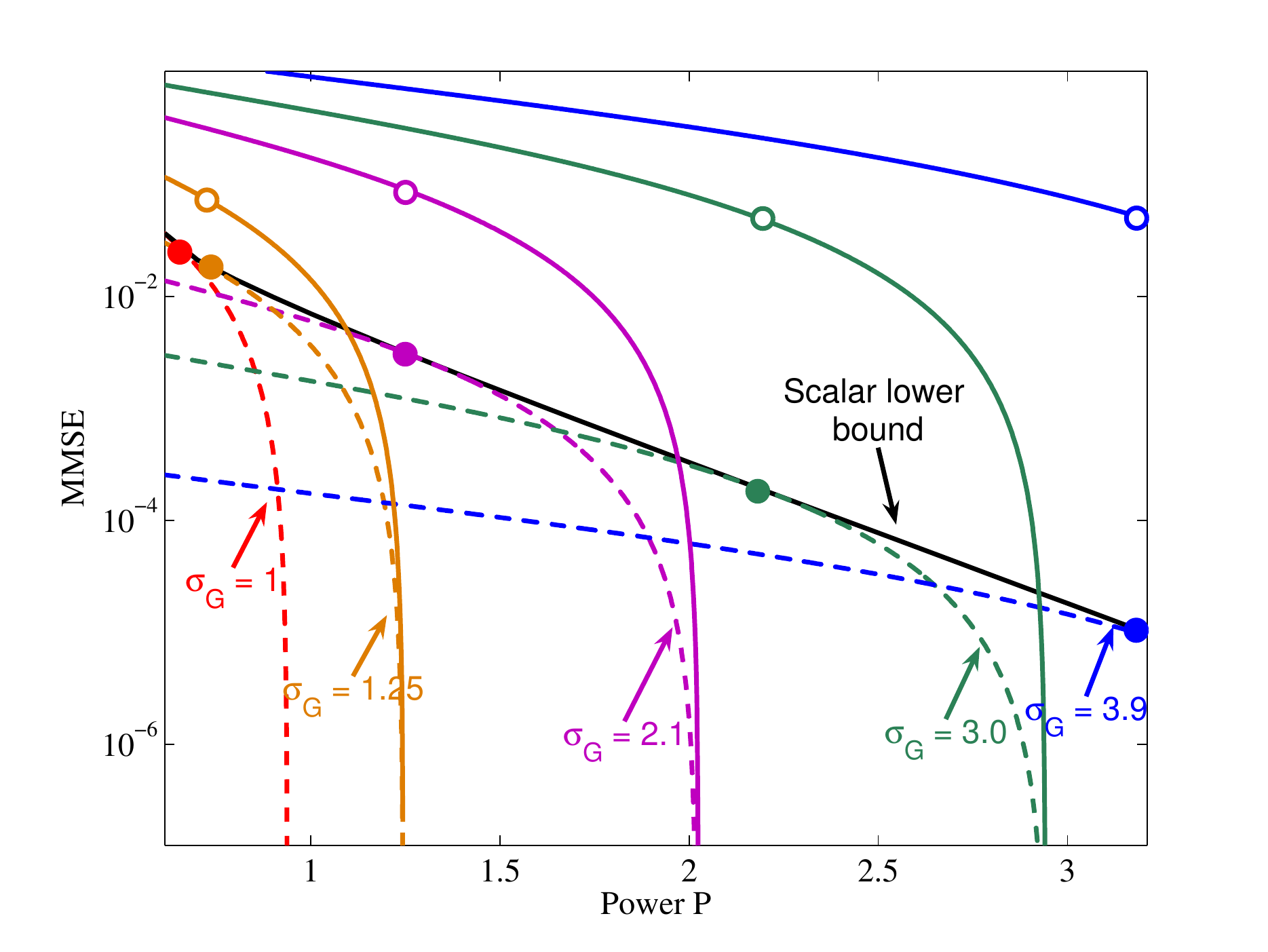}
\caption{A pictorial representation of the proof for the lower bound
  assuming $\sigma_0^2=30$. The solid curves show the vector lower bound
  of~\cite{WitsenhausenJournal} for various values of observation
  noise variances, denoted by $\sigma_G^2$. Conceptually, multiplying these curves by the
  probability of that channel behavior yields the shadow curves for the particular $\sigma_G^2$, shown by dashed curves. The scalar lower bound is
 then obtained by taking the maximum of these shadow curves.  The circles at points along the scalar bound curve indicate the
  optimizing value of $\sigma_G$ for obtaining that point on the
  bound.}
\label{fig:rainbow2}
\end{center}
\end{figure}

Bansal and Basar~\cite{bansalbasar} use information-theoretic techniques related to  rate-distortion and channel capacity to show the optimality of linear strategies in a modified version of Witsenhausen's counterexample where the cost function does not contain a product of two decision variables. Following the same spirit, in~\cite{WitsenhausenJournal} we derive the following lower bound for Witsenhausen's counterexample itself. 
\begin{theorem}
\label{thm:oldbound}
For $W(m,k^2,\sigma_0^2)$, if for a strategy $\gamma(\cdot{})$ the average power $\frac{1}{m}\expectp{\m{X}_0}{\|\m{U}_1\|^2}=P$, the following lower bound holds on the second stage cost 
\begin{equation*}
\bar{J}_2^{(\gamma)}(m,k^2,\sigma_0^2) \geq  \left(  \left(\sqrt{\kappa(P,\sigma_0^2)} - \sqrt{P}\right)^+   \right)^2,
\end{equation*}
where $(\cdot{})^+$ is shorthand for $\max(\cdot{}, 0)$ and 
\begin{equation}
\kappa(P,\sigma_0^2)=\frac{\sigma_0^2}{\sigma_0^2+P+2\sigma_0\sqrt{P}+1}.
\end{equation}
The following lower bound thus holds on the total cost
\begin{equation*}
\bar{J}^{(\gamma)}(m,k^2,\sigma_0^2) \geq \inf_{P\geq 0} k^2P + \left(  \left(\sqrt{\kappa(P,\sigma_0^2)} - \sqrt{P}\right)^+   \right)^2.
\end{equation*}
\end{theorem}
\begin{proof}
We refer the reader to~\cite{WitsenhausenJournal} for the full proof. We outline it here because these ideas are used in the derivation of the new lower bound in Theorem~\ref{thm:newbound}. 

Using a triangle inequality argument, we show 
\begin{eqnarray}
\sqrt{\frac{1}{m}\expectp{\m{X}_0,\m{Z}}{\|\m{X}_0-\whatmn{X}_1\|^2}}\leq \sqrt{\frac{1}{m}\expectp{\m{X}_0,\m{Z}}{\|\m{X}_0-\m{X}_1\|^2}}+\sqrt{\frac{1}{m}\expectp{\m{X}_0,\m{Z}}{\|\m{X}_1 - \whatmn{X}_1\|^2}}.
\label{eq:triangle1}
\end{eqnarray}
The first term on the RHS is $\sqrt{P}$. It therefore suffices to lower bound the term on the LHS to obtain a lower bound on  $\expectp{\m{X}_0,\m{Z}}{\|\m{X}_1 - \whatmn{X}_1\|^2}$. To that end, we interpret $\whatmn{X}_1$ as an estimate for $\m{X}_0$, which is a problem of transmitting a source across a channel. For an iid Gaussian source to be transmitted across  a memoryless power-constrained additive-noise Gaussian channel (with one channel use per source symbol), the optimal strategy that minimizes the mean-square error is merely scaling the source symbol so that the average power constraint is met~\cite{GoblickUncoded}. The estimation at the second controller is then merely the linear MMSE estimation of $\m{X}_0$, and the obtained MMSE is $\kappa(P,\sigma_0^2)$. The lemma now follows from~\eqref{eq:triangle1}.
\end{proof}

Observe that the lower bound expression is the same for all vector
lengths. In the following, large-deviation arguments~\cite{BlahutThesis,CsiszarKorner} (called sphere-packing style arguments for historical reasons) are extended following~\cite{PinskerNoFeedback,OurUpperBoundPaper,waterslide} to a joint source-channel setting where the distortion measure is unbounded. The obtained bounds are tighter than those in Theorem~\ref{thm:oldbound} and depend explicitly on the  vector length $m$.
\begin{theorem}
\label{thm:newbound}
For $W(m,k^2,\sigma_0^2)$,  if for a strategy $\gamma(\cdot{})$ the average power $\frac{1}{m}\expectp{\m{X}_0}{\|\m{U}_1\|^2}=P$, the following lower bound holds on the second stage cost for any choice of $\sigma_G^2\geq 1$ and $L>0$
\begin{equation*}
\bar{J}_2^{(\gamma)}(m,k^2,\sigma_0^2) \geq \eta(P,\sigma_0^2,\sigma_G^2,L).
\end{equation*}
where 
\begin{eqnarray*}
\eta(P,\sigma_0^2,\sigma_G^2,L)=\frac{\sigma_G^m}{c_m(L)}\exp\left(-\frac{mL^2(\sigma_G^2-1)}{2}\right)\left(  \left(\sqrt{\kappa_2(P,\sigma_0^2,\sigma_G^2,L)} - \sqrt{P}\right)^+   \right)^2,
\end{eqnarray*}
where $\kappa_2(P,\sigma_0^2,\sigma_G^2,L):=$
\begin{eqnarray*}
\frac{\sigma_0^2\sigma_G^2}{c_m^{\frac{2}{m}}(L)e^{1-d_m(L)}\left((\sigma_0+\sqrt{P})^2+d_m(L)\sigma_G^2\right)},
\end{eqnarray*}
 $c_m(L):=\frac{1}{\Pr(\|\m{Z}\|^2\leq mL^2)}= \left(1-\psi(m,L\sqrt{m})\right)^{-1}$, 
$d_m(L):=\frac{\Pr(\|\mk{Z}{m+2}\|^2\leq mL^2)}{\Pr(\|\m{Z}\|^2\leq mL^2)} =
\frac{1-\psi(m+2,L\sqrt{m})}{1-\psi(m,L\sqrt{m})}$, \\$0< d_m(L)<1$, and 
$\psi(m,r)=\Pr(\|\m{Z}\|\geq r)$.
Thus the following lower bound holds on the total cost
\begin{equation}
\bar{J}_{\min}(m,k^2,\sigma_0^2) \geq \inf_{P\geq 0} k^2P + 
\eta(P,\sigma_0^2,\sigma_G^2,L),
\end{equation}
for any choice of $\sigma_G^2\geq 1$ and $L>0$ (the choice can depend on $P$). Further, these bounds are at least as tight as those of Theorem~\ref{thm:oldbound} for all values of $k$ and $\sigma_0^2$.
\end{theorem}
\begin{proof}
From Theorem~\ref{thm:oldbound}, for a given $P$, a lower bound on the average second stage cost is $\left(\left( \sqrt{\kappa}-\sqrt{P}   \right)^+\right)^2$. We derive
 another lower bound that is equal to the 
 expression for $\eta(P,\sigma_0^2,\sigma_G^2,L)$. The high-level intuition behind this lower
 bound is presented in Fig.~\ref{fig:rainbow2}.

 Define $\mathcal{S}_L^G:=\{\m{z}:\|\m{z}\|^2\leq mL^2\sigma_G^2\}$
 and use subscripts to denote which probability model is being used for the second stage observation noise. $Z$ denotes white Gaussian
 of variance $1$ while $G$ denotes white Gaussian of variance
 $\sigma_G^2\geq 1$. 
\begin{eqnarray}
\nonumber\expectp{\m{X}_0,\m{Z}}{J_2^{(\gamma)}(\m{X}_0,\m{Z})}&= & \int_{\m{z}}\int_{\m{x}_0}J_2^{(\gamma)}(\m{x}_0,\m{z}) f_0(\m{x}_0) f_Z(\m{z}) d\m{x}_0 d\m{z}\\
\nonumber &\geq & \int_{\m{z}\in\mathcal{S}_L^G}\left(\int_{\m{x}_0}J_2^{(\gamma)}(\m{x}_0,\m{z}) f_0(\m{x}_0) d\m{x}_0\right) f_Z(\m{z}) d\m{z}\\
&= &\int_{\m{z}\in\mathcal{S}_L^G}\left(\int_{\m{x}_0}J_2^{(\gamma)}(\m{x}_0,\m{z}) f_0(\m{x}_0) d\m{x}_0\right)\frac{f_Z(\m{z})}{f_G(\m{z})}f_G(\m{z}) d\m{z}.
\label{eq:beforeratio}
\end{eqnarray}
The ratio of the two probability density functions is given by
\begin{eqnarray*}
\frac{f_Z(\m{z})}{f_G(\m{z})}=\frac{e^{-\frac{\|\m{z}\|^2}{2}}}{\left(\sqrt{2\pi}\right)^m}\frac{\left(\sqrt{2\pi\sigma_G^2}\right)^m}{e^{-\frac{\|\m{z}\|^2}{2\sigma_G^2}}}=\sigma_G^m e^{-\frac{\|\m{z}\|^2}{2}\left(1-\frac{1}{\sigma_G^2}\right)}.
\end{eqnarray*}
Observe that $\m{z}\in\mathcal{S}_L^G$, $\|\m{z}\|^2\leq mL^2\sigma_G^2$. Using $\sigma_G^2\geq 1$, we obtain
\begin{equation}
\frac{f_Z(\m{z})}{f_G(\m{z})}
\geq \sigma_G^m 
e^{-\frac{m L^2 \sigma_G^2}{2}\left(1-\frac{1}{\sigma_G^2}\right)}
= \sigma_G^m e^{-\frac{mL^2(\sigma_G^2-1)}{2}}.
\label{eq:afterratio}
\end{equation}
Using~\eqref{eq:beforeratio} and~\eqref{eq:afterratio},
\begin{eqnarray}
\nonumber\expectp{\m{X}_0,\m{Z}}{J_2^{(\gamma)}(\m{X}_0,\m{Z})} &\geq &\sigma_G^m e^{-\frac{mL^2(\sigma_G^2-1)}{2}} \int_{\m{z}\in\mathcal{S}_L^G}\left(\int_{\m{x}_0}J_2^{(\gamma)}(\m{x}_0,\m{z}) f_0(\m{x}_0) d\m{x}_0\right)  f_G(\m{z}) d\m{z}\\
\nonumber&=&\sigma_G^m e^{-\frac{mL^2(\sigma_G^2-1)}{2}}\expectp{\m{X}_0,\m{Z}_G}{J_2^{(\gamma)}(\m{X}_0,\m{Z}_G)\indi{\m{Z}_G\in\mathcal{S}_L^G}}\\
&=&\sigma_G^m
e^{-\frac{mL^2(\sigma_G^2-1)}{2}}\expectp{\m{X}_0,\m{Z}_G}{J_2^{(\gamma)}(\m{X}_0,\m{Z}_G)|\m{Z}_G\in\mathcal{S}_L^G}\Pr(\m{Z}_G\in\mathcal{S}_L^G).
\label{eq:explb}
\end{eqnarray}
Analyzing the probability term in~\eqref{eq:explb},
\begin{eqnarray}
\nonumber\Pr(\m{Z}_G\in\mathcal{S}_L^G)&=& \Pr\left(\|\m{Z}_G\|^2\leq mL^2\sigma_G^2\right)= \Pr\left(\left(\frac{\|\m{Z}_G\|}{\sigma_G}\right)^2\leq mL^2\right)\\
&=& 1-\Pr\left(\left(\frac{\|\m{Z}_G\|}{\sigma_G}\right)^2> mL^2\right)
= 1-\psi(m,L\sqrt{m}) = \frac{1}{c_m(L)},
\label{eq:sphereprob}
\end{eqnarray}
because $\frac{\m{Z}_G}{\sigma_G}\sim\mathcal{N}(0,\mathbb{I}_m)$. From~\eqref{eq:explb} and~\eqref{eq:sphereprob},
\begin{eqnarray}
\nonumber \expectp{\m{X}_0,\m{Z}}{J_2^{(\gamma)}(\m{X}_0,\m{Z})}&\geq & \sigma_G^m
e^{-\frac{mL^2(\sigma_G^2-1)}{2}}\expectp{\m{X}_0,\m{Z}_G}{J_2^{(\gamma)}(\m{X}_0,\m{Z}_G)|\m{Z}_G\in\mathcal{S}_L^G}(1-\psi(m,L\sqrt{m}))\\
& = & \frac{\sigma_G^m e^{-\frac{mL^2(\sigma_G^2-1)}{2}}}{c_m(L)}  \expectp{\m{X}_0,\m{Z}_G}{J_2^{(\gamma)}(\m{X}_0,\m{Z}_G)|\m{Z}_G\in\mathcal{S}_L^G}.
\label{eq:ep0z}
\end{eqnarray}
We now need the following lemma, which connects the new finite-length lower bound to the infinite-length lower bound
of~\cite{WitsenhausenJournal}.
\begin{lemma}
\label{lem:epg}
\begin{eqnarray*}
\expectp{\m{X}_0,\m{Z}_G}{J_2^{(\gamma)}(\m{X}_0,\m{Z}_G)|\m{Z}_G\in \mathcal{S}_L^G}
\geq \left(\left(   \sqrt{\kappa_2 (P,\sigma_0^2,\sigma_G^2,L)} -\sqrt{P}           \right)^+\right)^2,
\end{eqnarray*}
for any $L>0$.
\end{lemma}

\begin{proof}
See Appendix~\ref{app:ep0g}.
\end{proof} 
The lower bound on the total average cost now follows from~\eqref{eq:ep0z} and Lemma~\ref{lem:epg}. 

We now verify that $d_m(L)\in(0,1)$. That $d_m(L)>0$ is clear from definition. $d_m(L)<1$ because $\{\mk{z}{m+2}:\| \mk{z}{m+2}\|^2\leq mL^2\sigma_G^2\} \subset  \{\mk{z}{m+2}:  \|\mk{z}{m}\|^2\leq mL^2\sigma_G^2\}$, \emph{i.e.}, a sphere sits inside a cylinder.

Finally we verify that this new lower bound is at least as tight as the one in Theorem~\ref{thm:oldbound}. Choosing $\sigma_G^2=1$ in the expression for $\eta(P,\sigma_0^2,\sigma_G^2,L)$,
\begin{eqnarray*}
\eta(P,\sigma_0^2,\sigma_G^2,L)\geq \sup_{L>0}\frac{1}{c_m(L)}\left(\left(   \sqrt{\kappa_2 (P,\sigma_0^2,1,L)} -\sqrt{P}           \right)^+\right)^2.
\end{eqnarray*}
Now notice that $c_m(L)$ and $d_m(L)$ converge to $1$ as $L\rightarrow\infty$. Thus $\kappa_2(P,\sigma_0^2,1,L)\overset{L\rightarrow\infty}{\longrightarrow} \kappa(P,\sigma_0^2)$ and therefore, $\eta(P,\sigma_0^2,\sigma_G^2,L)$ is lower bounded by $\left(\left( \sqrt{\kappa}-\sqrt{P}   \right)^+\right)^2$, the lower bound in Theorem~\ref{thm:oldbound}.

\end{proof} 

\section{Combination of linear and lattice-based strategies attain within a constant factor of the optimal cost}
\label{sec:ratio}

\begin{theorem}[Constant-factor optimality]
The costs for $W(m,k^2,\sigma_0^2)$ are bounded as follows
\begin{eqnarray*}
\inf_{P\geq 0} \sup_{\sigma_G^2\geq 1,L>0} k^2P+\eta(P,\sigma_0^2,\sigma_G^2,L) \leq 
\bar{J}_{min}(m,k^2,\sigma_0^2)
 \leq 
\mu \left(\inf_{P\geq 0}  \sup_{\sigma_G^2\geq 1,L>0} k^2P+\eta(P,\sigma_0^2,\sigma_G^2,L)\right),
\end{eqnarray*}
where $\mu=100\xi^2$, $\xi$ is the packing-covering ratio of any lattice in $\mathbb{R}^m$, and $\eta(\cdot)$ is as defined in Theorem~\ref{thm:newbound}. For any $m$, $\mu<1600$. Further, depending on the $(m,k^2,\sigma_0^2)$
values, the  upper bound can be attained by lattice-based quantization
strategies or linear strategies. For $m=1$, a numerical calculation (MATLAB code available at~\cite{finiteWitsenahusenMatlabCode}) shows that $\mu<8$ (see Fig.~\ref{fig:scalar2}). 
\end{theorem}
\begin{proof}
\begin{figure}
\begin{center}
\includegraphics[width=9cm]{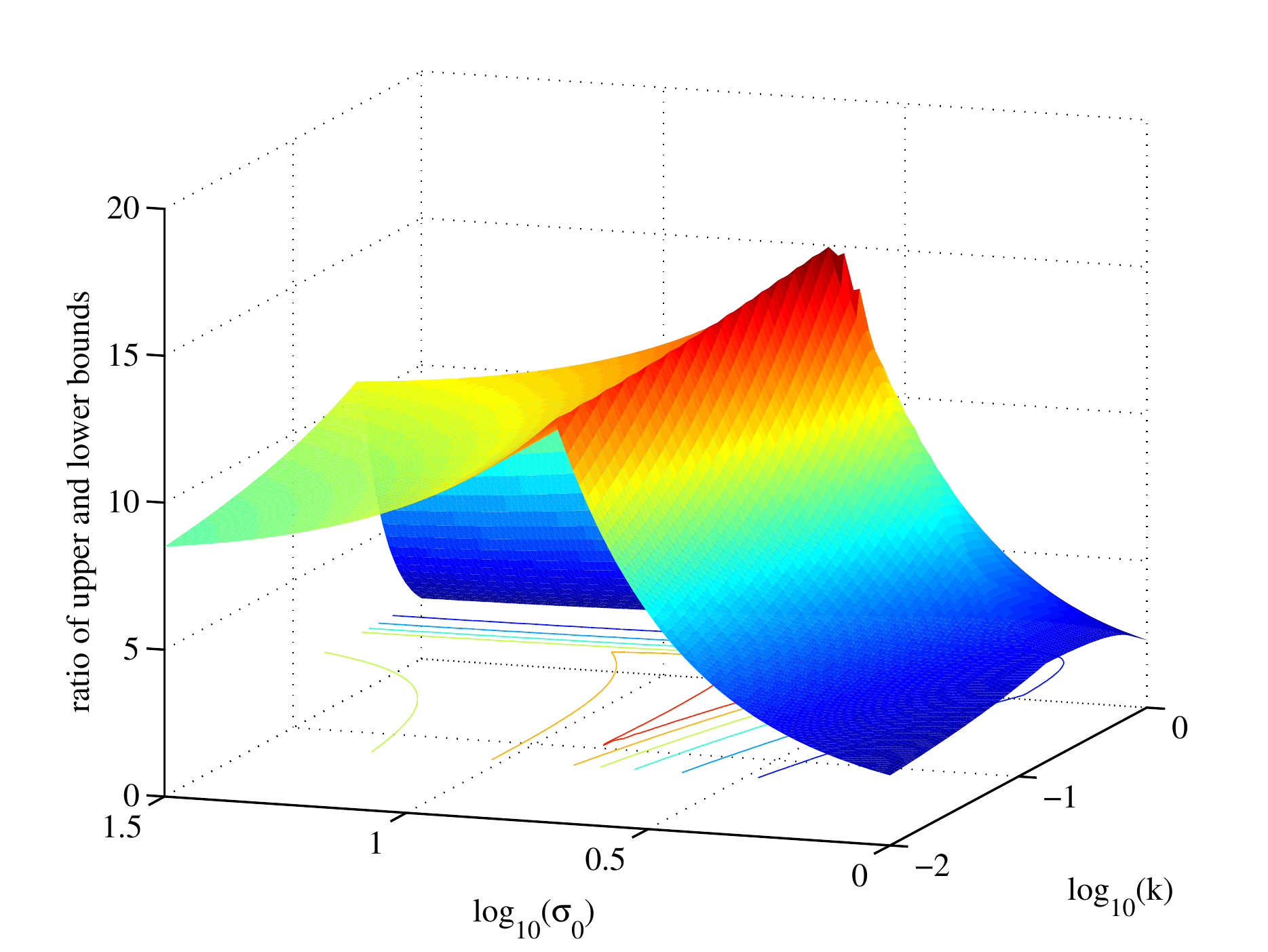}
\includegraphics[width=9cm]{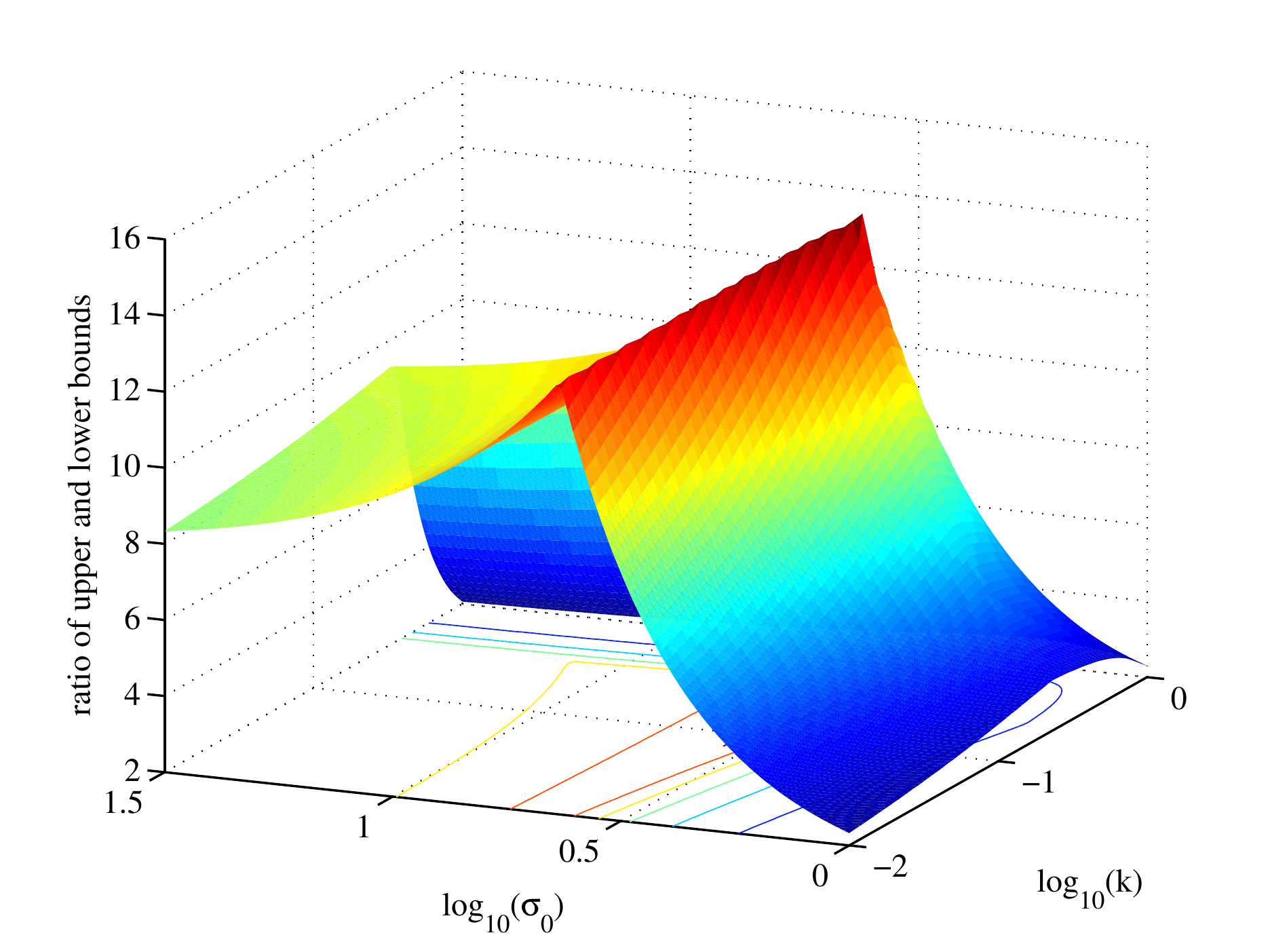}
\end{center}
\caption{The ratio of the upper and the lower bounds for the scalar
  Witsenhausen problem (top), and the 2-D Witsenhausen problem
  (bottom, using hexagonal lattice of $\xi=\frac{2}{\sqrt{3}}$) for a
  range of values of $k$ and $\sigma_0$. The ratio is bounded above by
  $17$ for the scalar problem, and by $14.75$ for the 2-D problem.} 
\label{fig:scalar}
\end{figure}
\begin{figure}
\begin{center}
\includegraphics[width=9cm]{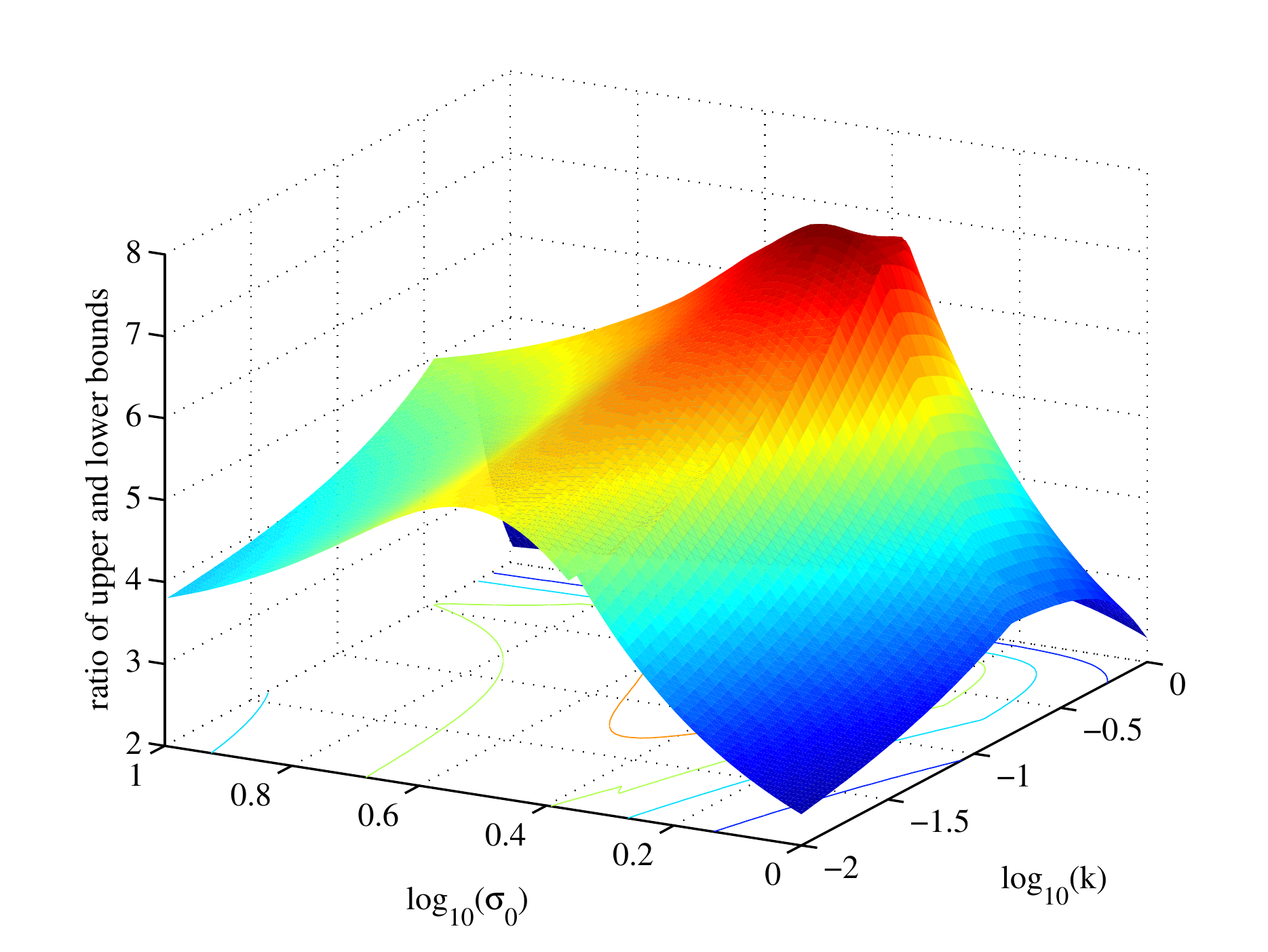}
\end{center}
\caption{An exact calculation of the first and second stage costs yields an improved maximum ratio smaller than $8$ for the scalar Witsenhausen problem.}
\label{fig:scalar2}
\end{figure}
Let $P^*$ denote the power $P$ in the lower bound in Theorem~\ref{thm:newbound}. We show  here that for any choice of $P^*$, the ratio of the upper and the lower bound is bounded. 

Consider the two simple linear strategies
of zero-forcing ($\m{u}_1=-\m{x}_0$) and zero-input ($\m{u}_1=0$)
followed by LLSE estimation at \co{2}. It is easy to
see~\cite{WitsenhausenJournal} that the average cost attained using these two
strategies is $k^2\sigma_0^2$ and $\frac{\sigma_0^2}{\sigma_0^2+1}<1$
respectively. An upper bound is obtained using
the best amongst the two linear strategies and the lattice-based
quantization strategy. 

\textit{Case 1}: $P^*\geq\frac{\sigma_0^2}{100}$. \\
The first stage cost is larger than $k^2\frac{\sigma_0^2}{100}$. Consider the upper bound of $k^2\sigma_0^2$ obtained by zero-forcing. The ratio of the upper bound and the lower bound is no larger than $100$. 


\textit{Case 2}: $P^*<\frac{\sigma_0^2}{100}$ and  $\sigma_0^2<16$.  \\
Using the bound from Theorem~\ref{thm:oldbound} (which is a special case of the bound in Theorem~\ref{thm:newbound}),
\begin{eqnarray*}
\kappa &=& \frac{\sigma_0^2}{(\sigma_0+\sqrt{P^*})^2+1}
\overset{\left(P^*<\frac{\sigma_0^2}{100}\right)}{\geq}  \frac{\sigma_0^2}{\sigma_0^2\left(1+\frac{1}{\sqrt{100}}\right)^2+1}\\
&\overset{(\sigma_0^2< 16)}{\geq} & \frac{\sigma_0^2}{16\left(1+\frac{1}{\sqrt{100}}\right)^2+1}=\frac{\sigma_0^2}{20.36 }\geq \frac{\sigma_0^2}{21}.
\end{eqnarray*}
Thus, for $\sigma_0^2<16$ and $P^*\leq \frac{\sigma_0^2}{100}$,
\begin{eqnarray*}
\bar{J}_{min}&\geq& \left((\sqrt{\kappa}-\sqrt{P^*})^+\right)^2\geq \sigma_0^2\left(\frac{1}{\sqrt{21}}-\frac{1}{\sqrt{100}}\right)^2\approx  0.014\sigma_0^2 \geq \frac{\sigma_0^2}{72}.
\end{eqnarray*}
Using the zero-input upper bound of $\frac{\sigma_0^2}{\sigma_0^2+1}$, the ratio of the upper and lower bounds is at most $\frac{72}{\sigma_0^2+1}\leq 72$.

\textit{Case 3}: $P^*\leq\frac{\sigma_0^2}{100}, \sigma_0^2\geq 16, P^*\leq \frac{1}{2}$.\\
In this case, 
\begin{eqnarray*}
\kappa &=&\frac{\sigma_0^2}{(\sigma_0+\sqrt{P^*})^2+1}\overset{(P^*\leq \frac{1}{2})}{\geq} \frac{\sigma_0^2}{(\sigma_0+\sqrt{0.5})^2+1}\\
&\overset{(a)}{\geq}& \frac{16}{(\sqrt{16}+\sqrt{0.5})^2+1}\approx 0.6909 \geq 0.69,
\end{eqnarray*}
where $(a)$ uses $\sigma_0^2\geq 16$ and the observation that $\frac{x^2}{(x+b)^2+1}=\frac{1}{\left(1+\frac{b}{x}\right)^2+\frac{1}{x^2}}$ is an increasing function of $x$ for $x,b>0$. Thus, 
\begin{eqnarray*}
\left((\sqrt{\kappa}-\sqrt{P})^+\right)^2\geq ((\sqrt{0.69}-\sqrt{0.5})^+)^2\approx 0.0153 \geq 0.015.
\end{eqnarray*}
Using the upper bound of $\frac{\sigma_0^2}{\sigma_0^2+1}<1$, the ratio of the upper and the lower bounds is smaller than $\frac{1}{0.015}<67$.

\textit{Case 4}: $\sigma_0^2>16$, $\frac{1}{2}<P^*\leq\frac{\sigma_0^2}{100}$

Using $L=2$ in the lower bound, 
\begin{eqnarray*}
c_m(L)&=&\frac{1}{\Pr(\|\m{Z}\|^2\leq mL^2)}=\frac{1}{1-\Pr(\|\m{Z}\|^2 > mL^2)}\\
& \overset{\text{(Markov's ineq.)}}{\leq} & \frac{1}{1-\frac{m}{mL^2}} \overset{(L=2)}{=}  \frac{4}{3},
\end{eqnarray*}
Similarly, 
\begin{eqnarray*}
d_m(2)& =& \frac{\Pr(\|\mk{Z}{m+2}\|^2\leq mL^2)}{\Pr(\|\m{Z}\|^2\leq mL^2)}\\
&\geq & \Pr(\|\mk{Z}{m+2}\|^2\leq mL^2) =  1-\Pr(\|\mk{Z}{m+2}\|^2> mL^2)\\
& \overset{\text{(Markov's ineq.)}}{\geq} & 1 - \frac{m+2}{mL^2} =  1 - \frac{1+\frac{2}{m}}{4}\overset{(m\geq 1)}{\geq} 1-\frac{3}{4}=\frac{1}{4}.
\end{eqnarray*}
In the bound, we are free to use any $\sigma_G^2\geq 1$. Using
$\sigma_G^2=6P^*>1$,
\begin{eqnarray*}
\kappa_2 &=&\frac{\sigma_G^2\sigma_0^2}{\left((\sigma_0+\sqrt{P^*})^2+d_m(2)\sigma_G^2\right)c_m^{\frac{2}{m}}(2)  e^{1-d_m(2)}     }\\
&\overset{(a)}{\geq}&\frac{6P^*\sigma_0^2}{\left((\sigma_0+\frac{\sigma_0}{10})^2+\frac{6\sigma_0^2}{100}\right)  \left(\frac{4}{3}\right)^{\frac{2}{m}}  e^{\frac{3}{4}}     }\overset{(m\geq 1)}{\geq}  1.255 P^*.
\end{eqnarray*}
where $(a)$ uses $\sigma_G^2=6P^*, P^*<\frac{\sigma_0^2}{100}, c_m(2)\leq \frac{4}{3}$ and $1>d_m(2)\geq \frac{1}{4}$. Thus, 
\begin{equation}
\left((\sqrt{\kappa_2}-\sqrt{P^*})^+\right)^2 \geq
P^*(\sqrt{1.255}-1)^2 \geq \frac{P^*}{70}.
\end{equation}
Now, using the lower bound on the total cost from Theorem~\ref{thm:newbound}, and substituting $L=2$,
\begin{eqnarray}
\nonumber \bar{J}_{min}(m,k^2,\sigma_0^2)  &\geq &
k^2P^* + 
\frac{\sigma_G^m}{c_m(2)}
\exp\left(-\frac{mL^2(\sigma_G^2-1)}{2}\right) \left(  \left(\sqrt{\kappa_2} - \sqrt{P^*}\right)^+ \right)^2\\
\nonumber &\overset{(\sigma_G^2=6P^*)}{\geq}& k^2P^* + \frac{(6P^*)^m}{c_m(2)} \exp\left( -\frac{4m(6P^*-1)}{2}   \right)\;\frac{P^*}{70} \\
\nonumber &\overset{(a)}{\geq}& k^2P^* + 
\frac{3^m}{\frac{4}{3}} e^{2m} e^{-12P^*m}\;\frac{1}{70\times 2}\\
\nonumber &\overset{(m\geq 1)}{\geq}& k^2P^* + \frac{3\times 3\times e^2}{4\times 70\times 2} e^{-12mP^*}\\
&> & k^2P^* + \frac{1}{9}e^{-12mP^*},
\label{eq:jminlower}
\end{eqnarray}
where $(a)$ uses $c_m(2) \leq \frac{4}{3}$ and $P^*\geq \frac{1}{2}$. We loosen the lattice-based upper bound from Theorem~\ref{thm:upperbound} and bring it into a form similar to~\eqref{eq:jminlower}. Here, $P$ is a part of the optimization:
\begin{eqnarray}
&&\bar{J}_{min}(m,k^2,\sigma_0^2)\nonumber\\
&\leq &\inf_{P>\xi^2}k^2P+\left(1+\sqrt{\frac{P}{\xi^2}}\right)^2e^{-\frac{mP}{2\xi^2}+\frac{m+2}{2}\left(1+\lon{\frac{P}{\xi^2}}  \right)}\nonumber\\
&\leq &\inf_{P>\xi^2}k^2P +\frac{1}{9}e^{-\frac{0.5mP}{\xi^2}+\frac{m+2}{2}\left(1+\lon{\frac{P}{\xi^2}}  \right) + 2\lon{   1+\sqrt{\frac{P}{\xi^2}   } }  +\lon{9}   }\nonumber\\
&\leq &\inf_{P>\xi^2}k^2P+\frac{1}{9}e^{-m\left(\frac{0.5P}{\xi^2}-\frac{m+2}{2m}\left(1+\lon{\frac{P}{\xi^2}}  \right) - \frac{2}{m}\lon{   1+\sqrt{\frac{P}{\xi^2}}    }   -\frac{\lon{9}}{m}   \right)}\nonumber\\
\nonumber&=&\inf_{P>\xi^2}k^2P+\frac{1}{9}e^{-\frac{0.12 mP}{\xi^2}}\times e^{-m\left(\frac{0.38P}{\xi^2}-\frac{1+\frac{2}{m}}{2}\left(1+\lon{\frac{P}{\xi^2}}\right)-\frac{2}{m}\lon{1+\sqrt{\frac{P}{\xi^2}}}  -\frac{\lon{9}}{m}   \right) } \nonumber\\
\nonumber&\overset{(m\geq 1)}{\leq}&\inf_{P>\xi^2}k^2P+\frac{1}{9}e^{-\frac{0.12 mP}{\xi^2}} e^{-m\left(\frac{0.38P}{\xi^2}-\frac{3}{2}\left(1+\lon{\frac{P}{\xi^2}}  \right)-2\lon{1+\sqrt{\frac{P}{\xi^2}}}  -\lon{9} \right)  } \nonumber \\
&\leq& \inf_{P\geq 34\xi^2} k^2P+\frac{1}{9}e^{-\frac{0.12 mP}{\xi^2}},
\label{eq:jminupper}
\end{eqnarray}
where the last inequality follows from the fact that
$\frac{0.38P}{\xi^2}>\frac{3}{2}\left(1+\lon{\frac{P}{\xi^2}} \right) +
2\lon{1+\sqrt{\frac{P}{\xi^2}}}+\lon{9}$ for $\frac{P}{\xi^2}>34$. This
can be checked easily by plotting it.\footnote{It can also be verified
symbolically by examining the expression $g(b) = 0.38b^2 -
\frac{3}{2}(1 + \ln b^2) - 2 \ln(1+b) - \lon{9}$, taking its derivative
$g'(b) = 0.76b - \frac{3}{b} - \frac{2}{1+b}$, and second
derivative $g''(b) = 0.76 + \frac{3}{b^2} + \frac{2}{(1+b)^2}
> 0$. Thus $g(\cdot{})$ is convex-$\cup$. Further, $g'(\sqrt{34})\approx 3.62>0$, and $g(\sqrt{34}) \approx 0.09$ and so $g(b) > 0$ whenever $b \geq \sqrt{34}$.}

Using $P=100\xi^2P^{*}\geq 50\xi^2>34\xi^2$ (since
$P^{*}\geq\frac{1}{2}$) in~\eqref{eq:jminupper},
\begin{eqnarray}
\nonumber
\bar{J}_{min}(m,k^2,\sigma_0^2)&\leq& k^2 100\xi^2P^{*}+\frac{1}{9}e^{-m\frac{0.12\times 100\xi^2P^{*}}{\xi^2}}\\
&=& k^2 100\xi^2P^{*}+\frac{1}{9}e^{-12mP^{*}}.
\label{eq:upper2}
\end{eqnarray} 
Using~\eqref{eq:jminlower} and~\eqref{eq:upper2}, the ratio of the
upper and the lower bounds is bounded for all $m$ since
\begin{equation}
\mu\leq \frac{ k^2 100\xi^2P^{*}+\frac{1}{9}e^{-12mP^{*}}}{
  k^2P^{*}+\frac{1}{9}e^{-12mP^{*}}}\leq \frac{ k^2 100\xi^2P^{*}}{
  k^2P^{*}}=100\xi^2. 
\end{equation}

For $m=1$, $\xi=1$, and thus in the proof the ratio $\mu\leq 100$. For $m$
large, $\xi\approx 2$~\cite{almosteverything}, and $\mu\lesssim
400$. For arbitrary $m$, using the recursive construction
in~\cite[Theorem 8.18]{Micciancio}, $\xi\leq 4$, and thus $\mu\leq
1600$ regardless of $m$. 
\end{proof}
Though the proof above succeeds in showing that the ratio is uniformly bounded by a constant, it is not very insightful and the constant is large. However, since the underlying vector bound  can be tightened (as shown in~\cite{ITW09Paper}), it is not worth improving the proof for increased elegance at this time. The important thing is that such a uniform constant exists. 

A numerical evaluation of the upper and lower bounds (of Theorem~\ref{thm:upperbound} and~\ref{thm:newbound} respectively) shows that the ratio is smaller than $17$ for $m=1$ (see Fig.~\ref{fig:scalar}). A precise calculation of the cost of the quantization strategy improves the upper bound to yield a maximum ratio smaller than $8$  (see Fig.~\ref{fig:scalar2}). 

A simple grid lattice has a packing-covering ratio $\xi=\sqrt{m}$. Therefore, while the grid lattice has the best possible packing-covering ratio of $1$ in the scalar case, it has a rather large packing covering ratio of $\sqrt{2}\;(\approx 1.41)$ for $m=2$. On the other hand, a hexagonal lattice (for $m=2$) has an improved packing-covering ratio of $\frac{2}{\sqrt{3}}\approx 1.15$. In contrast with $m=1$, where the ratio of upper and lower bounds of Theorem~\ref{thm:upperbound} and~\ref{thm:newbound} is approximately $17$, a hexagonal lattice yields a ratio smaller than $14.75$, despite having a larger packing-covering ratio. This is a consequence of the tightening of the sphere-packing lower bound (Theorem~\ref{thm:newbound}) as $m$ gets large\footnote{Indeed, in the limit $m\rightarrow\infty$, the ratio of the asymptotic average costs attained by a vector-quantization strategy and the vector lower bound of Theorem~\ref{thm:oldbound} is bounded by $4.45$~\cite{WitsenhausenJournal}.}.

\section{Discussions of numerical explorations and Conclusions}
\label{sec:conclusions}
Though lattice-based quantization strategies allow us to get within a constant
factor of the optimal cost for the vector Witsenhausen problem, they are
not optimal. This is known for the scalar~\cite{LeeLauHo} and the
infinite-length case~\cite{WitsenhausenJournal}. It is shown
in~\cite{WitsenhausenJournal} that the ``slopey-quantization" strategy of Lee, Lau and Ho~\cite{LeeLauHo} that is believed to be very close to optimal in the scalar case can be viewed as an instance of a linear scaling followed by a dirty-paper coding (DPC) strategy. Such DPC-based strategies are also the best known strategies in the asymptotic infinite-dimensional
case, requiring optimal power $P$ to attain $0$ asymptotic mean-square error in the estimation of $\m{x}_1$, and attaining costs within a factor of $1.3$ of the optimal~\cite{ITW09Paper} for all $(k,\sigma_0^2)$. This leads us to conjecture that a DPC-like strategy might be optimal for finite-vector lengths as well. In the following, we numerically explore the performance of DPC-like strategies.

\begin{figure}
\begin{center}
\includegraphics[width = 4 in]{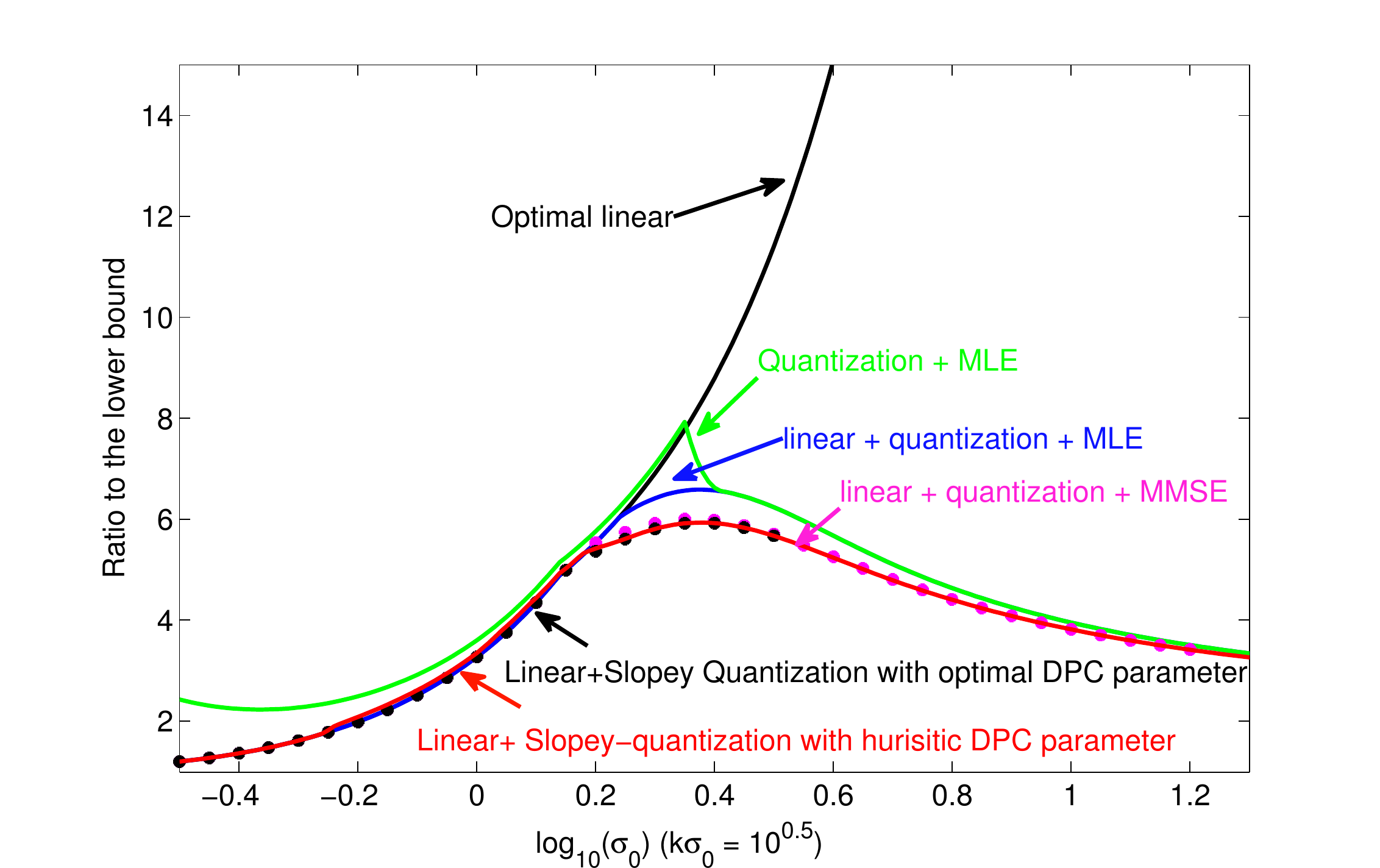}
\end{center}
\caption{Ratio of the achievable costs to the scalar lower bound along $k\sigma_0 =10^{-0.5}$ for various strategies. Quantization with MMSE-estimation at the second controller outperforms quantization with MLE, or even scaled MLE. For slopey-quantization with heuristic DPC-parameter, the parameter $\alpha$ in DPC-based scheme is borrowed from the infinite-length analysis. The figure suggests that along this path ($k\sigma_0=\sqrt{10}$), the difference between optimal-DPC and heuristic DPC is not substantial. However, Fig.~\ref{fig:abvsabc} (b) shows that this is not true in general. }
\label{fig:line}
\end{figure}

\begin{figure}
\begin{center}
\includegraphics[width = 6 in]{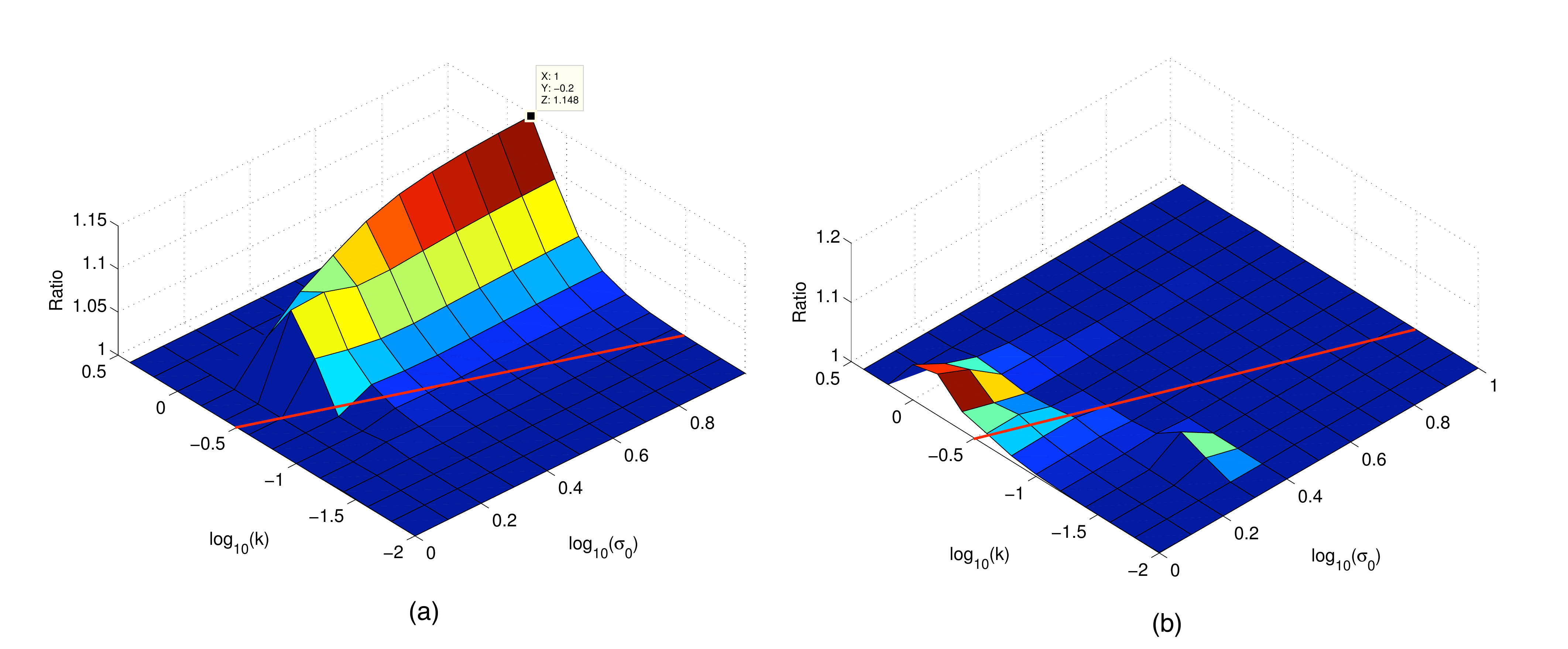}
\end{center}
\caption{(a) shows the ratio of cost attained by linear+quantization (with MMSE decoding) to DPC with parameter $\alpha$ obtained by brute-force optimization. DPC can do up to $15\%$ better than the optimal quantization strategy. Also the maximum is attained along $k\approx 0.6$ which is different from $k=0.2$ of the benchmark problem~\cite{LeeLauHo}. (b) shows the ratio of cost attained by linear+quantization to DPC with $\alpha$ borrowed from infinite-length optimization. Heuristic DPC does not outperform linear+quantization (with MMSE estimation) substantially. }
\label{fig:abvsabc}
\end{figure}

%
It is natural to ask how much there is to gain using a DPC-based strategy over a simple quantization strategy. Notice that the DPC-strategy gains not only from the slopey quantization, but also from the MMSE-estimation at the second controller. In Fig.~\ref{fig:line}, we eliminate the latter advantage by considering first a uniform quantization-based strategy with an appropriate scaling of the MLE so that it approximates the MMSE-estimation performance, and then the actual MMSE-estimation strategy for uniform quantization. Along the curve $k\sigma_0=\sqrt{10}$, there is significant gain in using this approximate-MMSE estimation over MLE, and further gain in using MMSE-estimation itself. This also shows that there is an interesting tradeoff between the  complexity of the second controller and the system performance.

From Fig.~\ref{fig:line}, along the curve $k\sigma_0=\sqrt{10}$, the DPC-based strategy performs only negligibly better than a quantization-based strategy with MMSE estimation. Fig.~\ref{fig:abvsabc} (a) shows that this is not true in general. A DPC-based strategy can perform up to $15\%$ better than a simple quantization-based scheme depending on the problem parameters. Interestingly, the advantage of using a DPC-based strategy for the case of $k=0.2,\sigma_0=5$ (which is used as the benchmark case in many papers, e.g.~\cite{LeeLauHo,marden}) is quite small. The maximum gain of about $15\%$ is obtained at $k\approx 10^{-0.2}\approx 0.63$, and $\sigma_0=1$ (and indeed, any $\sigma_0>1$.  In the future, we suggest the community use the point $(0.63,1)$ as the benchmark case. 

Given that there is an advantage in using a DPC-like strategy, an interesting question is whether the DPC parameter $\alpha$ that optimizes the DPC-based strategy's performance at infinite-lengths (in~\cite{WitsenhausenJournal}) gives good performance for the scalar case as well. Fig.~\ref{fig:abvsabc} (b) answers this question at least partially in the negative. This heuristic-DPC does only slightly better than a quantization strategy with MMSE estimation, whereas other values of $\alpha$ do significantly better.

Finally, we observe that while uniform bin-size quantization or DPC-based strategies are designed for atypical noise behavior, atypical behavior of the the initial state is better accommodated by using nonuniform bin-sizes (such as those in~\cite{LeeLauHo,marden}). Table~\ref{tbl:yuchiho} compares the two. Clearly, the advantage in having nonuniform slopey-quantization is small, but not negligible. It would be interesting to calibrate the advantage of nonuniform-bin sizes for $(k,\sigma_0)=(0.63,1)$, a maximum gain point for uniform-bin size slopey-quantization strategies. 

\begin{table}[h!b!p!]
\caption{Costs attained for the benchmark case of $k=0.2$, $\sigma_0=5$.}
\begin{center}
\begin{tabular}{|c|c|c|}
  \hline
   & linear+quantization & Slopey-quantization \\
\hline
Lee, Lau and Ho  \cite{LeeLauHo} & 0.1713946 & 0.1673132 \\
Li, Marden and Shamma~\cite{marden} & --- & 0.1670790\\
  This paper & 0.1715335 & 0.1673654 \\
  \hline
\end{tabular}
\end{center}
  \label{tbl:yuchiho}
\end{table}


There are plenty of open problems that arise naturally. Both the lower and the upper bounds have room for improvement. The lower bound can be improved by tightening the vector lower bound of~\cite{WitsenhausenJournal} (one such tightening is performed in~\cite{ITW09Paper}) and obtaining corresponding finite-length results using the sphere-packing tools developed here. 

Tightening the upper bound can be performed by using  DPC-based techniques over lattices. Further, an exact analysis of the required
first-stage power when using a lattice would yield an improvement (as
pointed out earlier, for $m=1$, $\frac{1}{m}k^2r_c^2$ overestimates
the required first-stage cost), especially for small $m$. Improved
lattice designs with better packing-covering ratios would also improve
the upper bound.


Perhaps a more significant set of open problems are the next steps in understanding more realistic versions of Witsenhausen's problem,
specifically those that include costs on all the inputs and all the states~\cite{Allerton09Paper}, with noisy state evolution and noisy observations at both controllers. The hope is that solutions to these problems can then be used as the basis for provably-good nonlinear controller synthesis for larger distributed systems. Further, tools developed for solving these problems might help address multiuser problems in information theory, in the spirit of~\cite{WuInterferenceControl,EliaPaper1}.

\section*{Acknowledgments}
We gratefully acknowledge the support of the National Science Foundation (CNS-403427, CNS-093240, CCF-0917212 and CCF-729122), Sumitomo Electric and Samsung. We thank Amin Gohari, Bobak Nazer and Anand Sarwate for helpful discussions, and Gireeja Ranade for suggesting improvements in the paper. 

\appendices{}
\section{Proof of Lemma~\ref{lem:upperbound}}
\label{app:upperbound}
\begin{eqnarray}
\nonumber \expectp{\m{Z}}{\left(\|\m{Z}\|+r_p\right)^2\indi{\mathcal{E}_m}}&=&
\expectp{\m{Z}}{\|\m{Z}\|^2\indi{\mathcal{E}_m}}+r_p^2\Pr(\mathcal{E}_m)+2r_p\expectp{\m{Z}}{\left(\indi{\mathcal{E}_m}\right)
  \left(\|\m{Z}\|\indi{\mathcal{E}_m}\right)}\\ 
\nonumber&\overset{(a)}{\leq} &
\expectp{\m{Z}}{\|\m{Z}\|^2\indi{\mathcal{E}_m}}+r_p^2\Pr(\mathcal{E}_m) +2r_p\sqrt{\expectp{\m{Z}}{\indi{\mathcal{E}_m}}}
\sqrt{\expectp{\m{Z}}{\|\m{Z}\|^2\indi{\mathcal{E}_m}}}\\  
&=& \left( \sqrt{\expectp{\m{Z}}{\|\m{Z}\|^2\indi{\mathcal{E}_m}}  } + r_p\sqrt{ \Pr(\mathcal{E}_m)  }     \right)^2 ,
\label{eq:zplusrp}
\end{eqnarray}
where $(a)$ uses the Cauchy-Schwartz inequality~\cite[Pg. 13]{durrett}.


We wish to express $\expectp{\m{Z}}{\|\m{Z}\|^2\indi{\mathcal{E}_m}}$ in
terms of $\psi(m,r_p):=\Pr(\|\m{Z}\|\geq r_p)=\int_{\|\m{z}\|\geq
  r_p}\frac{e^{-\frac{\|\m{z}\|^2}{2}}}{\left(\sqrt{2\pi}\right)^m}d\m{z}$. 

Denote by $\mathcal{A}_m(r):=\frac{2\pi^{\frac{m}{2}} r^{m-1}
}{\Gamma\left(\frac{m}{2}\right)}$ the surface area of a sphere of
radius $r$ in $\mathbb{R}^m$~\cite[Pg. 458]{Courant}, where
$\Gamma(\cdot{})$ is the Gamma-function satisfying
$\Gamma(m)=(m-1)\Gamma(m-1)$, $\Gamma(1)=1$, and
$\Gamma(\frac{1}{2})=\sqrt{\pi}$. Dividing the space $\mathbb{R}^m$
into shells of thickness $dr$ and radii $r$, 
\begin{eqnarray}
\nonumber\expectp{\m{Z}}{\|\m{Z}\|^2\indi{\mathcal{E}_m}}&=&\int_{\|\m{z}\|\geq r_p}\|\m{z}\|^2\frac{e^{-\frac{\|\m{z}\|^2}{2}}}{\left(\sqrt{2\pi}\right)^m}d\m{z}= \int_{r\geq r_p}r^2\frac{e^{-\frac{r^2}{2}}}{\left(\sqrt{2\pi}\right)^m}\mathcal{A}_m(r)dr\\
\nonumber&=& \int_{r\geq r_p}r^2\frac{e^{-\frac{r^2}{2}}}{\left(\sqrt{2\pi}\right)^m}\frac{2\pi^{\frac{m}{2}} r^{m-1}  }{\Gamma\left(\frac{m}{2}\right)}dr\\
&=& \int_{r\geq r_p} \frac{e^{-\frac{r^2}{2}}2\pi}{\left(\sqrt{2\pi}\right)^{m+2}} \frac{2 \pi^{\frac{m+2}{2}} r^{m+1} }{\pi\frac{2}{m}\Gamma\left(\frac{m+2}{2}\right)}dr= m\psi(m+2,r_p).
\label{eq:psinplus2}
\end{eqnarray}
Using~\eqref{eq:zplusrp},~\eqref{eq:psinplus2}, and $r_p=\sqrt{\frac{mP}{\xi^2}}$
\begin{eqnarray*}
\nonumber\expectp{\m{Z}}{\left(\|\m{Z}\|+r_p\right)^2\indi{\mathcal{E}_m}}
\leq m\left(\sqrt{\psi(m+2,r_p)}+\sqrt{\frac{P}{\xi^2}}\sqrt{\psi(m,r_p)}\right)^2,
\end{eqnarray*}
which yields the first part of Lemma~\ref{lem:upperbound}. To obtain a
closed-form upper bound we consider $P>\xi^2$. It suffices to bound $\psi(\cdot{},\cdot{})$. 
\begin{eqnarray*}
\psi(m,r_p)&=&\Pr(\|\m{Z}\|^2\geq r_p^2)= \Pr(\exp(\rho\sum_{i=1}^mZ_i^2)\geq \exp(\rho r_p^2))\\
&\overset{(a)}{\leq} & \expectp{\m{Z}}{\exp(\rho\sum_{i=1}^mZ_i^2)}e^{-\rho r_p^2}
=\expectp{Z_1}{\exp(\rho Z_1^2)}^me^{-\rho r_p^2}
\overset{(\text{for}\;0<\rho<0.5)}{=} \frac{1}{(1-2\rho)^{\frac{m}{2}}}e^{-\rho r_p^2},
\end{eqnarray*}
where $(a)$ follows from the Markov inequality, and the last inequality follows from the fact that the moment generating function of a standard $\chi_2^2$ random variable is $\frac{1}{(1-2\rho)^{\frac{1}{2}}}$ for $\rho\in (0,0.5)$~\cite[Pg. 375]{ross}. Since this bound  holds for any $\rho\in (0,0.5)$, we choose the minimizing $\rho^*=\frac{1}{2}\left(1-\frac{m}{r_p^2}\right)$. Since $r_p^2=\frac{mP}{\xi^2}$, $\rho^*$ is indeed in $(0,0.5)$ as long as $P>\xi^2$. Thus,
\begin{eqnarray*}
\psi(m,r_p) \leq  \frac{1}{(1-2\rho^*)^{\frac{m}{2}}}e^{-\rho^* r_p^2}=  \left(\frac{r_p^2}{m}\right)^{\frac{m}{2}} e^{-\frac{1}{2}\left( 1-\frac{m}{r_p^2}  \right) r_p^2}= e^{-\frac{r_p^2}{2}+\frac{m}{2}+\frac{m}{2}\lon{\frac{r_p^2}{m}}}.
\end{eqnarray*}
Using the substitutions $r_c^2=mP$, $\xi=\frac{r_c}{r_p}$ and $r_p^2=\frac{mP}{\xi^2}$,
\begin{eqnarray}
\label{eq:psinrp}
 \Pr(\mathcal{E}_m)=\psi(m,r_p)=\psi\left(m,\sqrt{\frac{mP}{\xi^2}}\right)\leq e^{-\frac{mP}{2\xi^2}+\frac{m}{2}+\frac{m}{2}\lon{\frac{P}{\xi^2}}}, \;\text{and}
\end{eqnarray}
\begin{eqnarray}
\label{eq:psin2rp}
\expectp{\m{Z}}{\|\m{Z}\|^2\indi{\mathcal{E}_m}}\leq m\psi\left(m+2,\sqrt{\frac{mP}{\xi^2}}\right)\leq me^{-\frac{mP}{2\xi^2}+\frac{m+2}{2}+\frac{m+2}{2}\lon{\frac{mP}{(m+2)\xi^2}}}.
\end{eqnarray}
From~\eqref{eq:zplusrp},~\eqref{eq:psinrp} and~\eqref{eq:psin2rp}, 
\begin{eqnarray*}
\expectp{\m{Z}}{\left(\|\m{Z}\|+r_p\right)^2\indi{\mathcal{E}_m}}&\leq &\bigg(\sqrt{m}e^{-\frac{mP}{4\xi^2}+\frac{m+2}{4}+\frac{m+2}{4}\lon{\frac{mP}{(m+2)\xi^2}}}    \sqrt{\frac{mP}{\xi^2}}   e^{-\frac{mP}{4\xi^2}+\frac{m}{4}+\frac{m}{4}\lon{\frac{P}{\xi^2}}}\bigg)^2\\
&\overset{(\text{since}\;P>\xi^2)}{<} & \left(\sqrt{m}\left(1+\sqrt{\frac{P}{\xi^2}}\right)e^{-\frac{mP}{4\xi^2}+\frac{m+2}{4}+\frac{m+2}{4}\lon{\frac{P}{\xi^2}}}   \right)^2\\
&= & m\left(1+\sqrt{\frac{P}{\xi^2}}\right)^2e^{-\frac{mP}{2\xi^2}+\frac{m+2}{2}+\frac{m+2}{2}\lon{\frac{P}{\xi^2}}}.
\end{eqnarray*}

\section{Proof of Lemma~\ref{lem:epg}}
\label{app:ep0g}
The following lemma is taken from~\cite{WitsenhausenJournal}.
\begin{lemma}
For any three random variables $A$, $B$ and $C$,
\begin{eqnarray*}
\expect{\|B-C\|^2}\geq \left(\left(\sqrt{\expect{\|A-C\|^2}}-\sqrt{\expect{\|A-B\|^2}}\right)^{+}\right)^2.
\end{eqnarray*}
\end{lemma}
\begin{proof}
See~\cite[Appendix II]{WitsenhausenJournal}.
\end{proof}
Choosing $A=\m{X}_0$, $B=\m{X}_1$ and $C=\whatmn{X}_1$, 
\begin{eqnarray}
\nonumber && \expectp{\m{X}_0,\m{Z}_G}{J_2^{(\gamma)}(\m{X}_0,\m{Z}_G)|\m{Z}_G\in\mathcal{S}_L^G} =\frac{1}{m}\expectp{\m{X}_0,\m{Z}_G}{\|\m{X}_1-\whatmn{X}_1\|^2|\m{Z}_G\in\mathcal{S}_L^G}\\
\nonumber&\geq &\bigg(\bigg(\sqrt{\frac{1}{m}\expectp{\m{X}_0,\m{Z}_G}{\|\m{X}_0-\whatmn{X}_1\|^2|\m{Z}_G\in \mathcal{S}_L^G}} - \sqrt{\frac{1}{m}\expectp{\m{X}_0,\m{Z}_G}{\|\m{X}_0-\m{X}_1\|^2|\m{Z}_G\in \mathcal{S}_L^G}}  \bigg)^+\bigg)^2\\
 &= &\bigg(\bigg(\sqrt{\frac{1}{m}\expectp{\m{X}_0,\m{Z}_G}{\|\m{X}_0-\whatmn{X}_1\|^2|\m{Z}_L\in \mathcal{S}_L^G}}- \sqrt{P} \bigg)^+\bigg)^2,
\label{eq:sqrtd}
\end{eqnarray}
since $\m{X}_0-\m{X}_1=\m{U}_1$ is independent of $\m{Z}_G$ and
$\expect{\|\m{U}_1\|^2}= mP$. Define $\m{Y}_L:=\m{X}_1+\m{Z}_L$ to
be the output when the observation noise $\m{Z}_L$ is distributed as a truncated Gaussian distribution:
\begin{equation}
\label{eq:fz}
f_{Z_L}(\m{z}_L)=\left\{\begin{array}{ll}c_m(L)\frac{e^{-\frac{\|\m{z}_L\|^2}{2\sigma_G^2}}}{\left(\sqrt{2\pi\sigma_G^2}\right)^m}&\m{z}_L\in\mathcal{S}_L^G\\
0& \text{otherwise.}\end{array}\right.
\end{equation}
Let the estimate at the second controller on observing $\m{y}_L$ be
denoted by $\whatmn{X}_L$. Then, by the definition of conditional
expectations, 
\begin{equation}
\label{eq:xl}
\expectp{\m{X}_0,\m{Z}_G}{\|\m{X}_0-\whatmn{X}_1\|^2|\m{Z}_G\in \mathcal{S}_L^G} = \expectp{\m{X}_0,\m{Z}_G}{\|\m{X}_0-\whatmn{X}_L\|^2}.
\end{equation}
To get a lower bound, we now allow the controllers to optimize
themselves with the additional knowledge that the observation noise $\m{z}$ 
must fall in $\mathcal{S}_L^G$. In order to prevent the first controller from ``cheating'' and allocating different powers to the two events (\emph{i.e.} $\m{z}$ falling or not falling in $\mathcal{S}_L^G$), we enforce the constraint that the power $P$ must not change with this additional knowledge. Since the controller's observation $\m{X}_0$ is independent of $\m{Z}$, this constraint is satisfied by the original controller (without the additional knowledge) as well, and hence the cost for the system with the additional knowledge is still a valid lower bound to that of the original system.

The rest of the proof uses ideas from channel coding and the rate-distortion theorem~\cite[Ch. 13]{CoverThomas} from information theory. We view the problem as a problem of implicit communication from the first controller to the second. Notice that for a given $\gamma(\cdot{})$, $\m{X}_1$ is a function of $\m{X}_0$, $\m{Y}_L=\m{X}_1+\m{Z}_L$ is conditionally independent of $\m{X}_0$ given $\m{X}_1$ (since the noise $\m{Z}_L$ is additive and independent of $\m{X}_1$ and $\m{X}_0$). Further, $\whatmn{X}_L$ is a function of $\m{Y}_L$. Thus $\m{X}_0-\m{X}_1-\m{Y}_L-\whatmn{X}_L$ form a Markov chain. Using the
data-processing inequality~\cite[Pg. 33]{CoverThomas},
\begin{equation}
\label{eq:dpi}
I(\m{X}_0;\whatmn{X}_L)\leq I(\m{X}_1;\m{Y}_L),
\end{equation}
where $I(A,B)$ is the expression for mutual information expression between two random variables $A$ and $B$ (see, for example,~\cite[Pg. 18, Pg. 231]{CoverThomas}). To estimate the distortion to which $\m{X}_0$ can be communicated across this truncated Gaussian channel (which, in turn, helps us lower bound the MMSE in estimating $\m{X}_1$), we need to upper bound the term on the RHS of~\eqref{eq:dpi}. 
\begin{lemma}
\label{lem:capacity}
\begin{equation*} 
\frac{1}{m}I(\m{X}_1;\m{Y}_L) \leq  \frac{1}{2}\lo{\frac{e^{1-d_m(L)} (\bar{P}+d_m(L)\sigma_G^2)c_m^{\frac{2}{m}}(L)}{\sigma_G^2} }.
\end{equation*}
\end{lemma}
\begin{proof}
We first obtain an upper bound to the power of $\m{X}_1$ (this bound is the same as that used in~\cite{WitsenhausenJournal}):
\begin{eqnarray*}
\expectp{\m{X}_0}{\|\m{X}_1\|^2}&=&\expectp{\m{X}_0}{\|\m{X}_0+\m{U}_1\|^2}=\expectp{\m{X}_0}{\|\m{X}_0\|^2}+\expectp{\m{X}_0}{\|\m{U}_1\|^2}+2\expectp{\m{X}_0}{{\m{X}_0}^T\m{U}_1}\\
&\overset{(a)}{\leq} & \expectp{\m{X}_0}{\|\m{X}_0\|^2}+\expectp{\m{X}_0}{\|\m{U}_1\|^2}
+2\sqrt{\expectp{\m{X}_0}{\|\m{X}_0\|^2}}\sqrt{\expectp{\m{X}_0}{\| \m{U}_1\|^2}}\\
&\leq & m(\sigma_0+\sqrt{P})^2,
\end{eqnarray*}
where $(a)$ follows from the Cauchy-Schwartz inequality. We use the following definition of \textit{differential entropy} $h(A)$ of a continuous random variable $A$~\cite[Pg. 224]{CoverThomas}:
\begin{equation}
h(A) = -\int_S f_A(a) \lo{f_A(a)} da, 
\end{equation}
where $f_A(a)$ is the pdf of $A$, and $S$ is the support set of $A$. Conditional differential entropy is defined similarly~\cite[Pg. 229]{CoverThomas}. 

Let $\bar{P}:=(\sigma_0+\sqrt{P})^2$. Now, $\expect{Y_{L,i}^2}  = \expect{X_{1,i}^2} + \expect{Z_{L,i}^2} $ (since $X_{1,i}$ is independent of $Z_{L,i}$ and by symmetry, $Z_{L,i}$ are zero mean random variables). Denote $\bar{P}_i=\expect{X_{1,i}^2}$ and $\sigma_{G,i}^2=\expect{Z_{L,i}^2}$. In the following, we derive an upper bound $C_{G,L}^{(m)}$ on $\frac{1}{m}I(\m{X}_1;\m{Y}_L)$. 
\begin{eqnarray}
\nonumber C_{G,L}^{(m)}&:=&\sup_{p(\m{X}_1):\expect{\|\m{X}_1\|^2}\leq m\bar{P}}\frac{1}{m}I(\m{X}_1;\m{Y}_L)\\
\nonumber &\overset{(a)}{=}&\sup_{p(\m{X}_1):\expect{\|\m{X}_1\|^2}\leq m\bar{P}}\frac{1}{m}h(\m{Y}_L)-\frac{1}{m}h(\m{Y}_L|\m{X}_1)\\
\nonumber &   \overset{}{=}  &  \sup_{p(\m{X}_1):\expect{\|\m{X}_1\|^2}\leq m\bar{P}}\frac{1}{m}h(\m{Y}_L)-\frac{1}{m}h(\m{X}_1+\m{Z}_L|\m{X}_1)\\
\nonumber &   \overset{(b)}{=}  &  \sup_{p(\m{X}_1):\expect{\|\m{X}_1\|^2}\leq m\bar{P}}\frac{1}{m}h(\m{Y}_L)-\frac{1}{m}h(\m{Z}_L|\m{X}_1)\\
\nonumber &   \overset{(c)}{=}  &  \sup_{p(\m{X}_1):\expect{\|\m{X}_1\|^2}\leq m\bar{P}}\frac{1}{m}h(\m{Y}_L)-\frac{1}{m}h(\m{Z}_L)\\
\nonumber &\overset{(d)}{\leq} &\sup_{p(\m{X}_1):\expect{\|\m{X}_1\|^2}\leq m\bar{P}}\frac{1}{m}\sum_{i=1}^mh(Y_{L,i})-\frac{1}{m}h(\m{Z}_L)\\
\nonumber &\overset{(e)}{\leq} &\sup_{\bar{P}_i:\sum_{i=1}^m\bar{P}_i \leq m\bar{P}} \frac{1}{m}\sum_{i=1}^m\frac{1}{2}\lo{2\pi e(\bar{P}_i+\sigma_{G,i}^2)}-\frac{1}{m}h(\m{Z}_L)\\
&\overset{(f)}{\leq} & \frac{1}{2}\lo{2\pi e (\bar{P}+d_m(L)\sigma_G^2)}-\frac{1}{m}h(\m{Z}_L).
\label{eq:cn}
\end{eqnarray}
Here, $(a)$ follows from the definition of mutual information~\cite[Pg. 231]{CoverThomas}, $(b)$ follows from the fact that translation does not change the differential entropy~\cite[Pg. 233]{CoverThomas}, $(c)$ uses independence of $\m{Z}_L$ and $\m{X}_1$, and $(d)$ uses the chain rule for differential entropy~\cite[Pg. 232]{CoverThomas} and the fact that conditioning reduces entropy~\cite[Pg. 232]{CoverThomas}. In $(e)$, we used the fact that Gaussian random variables maximize
differential entropy. The inequality $(f)$ follows from the concavity-$\cap$ of the $\log(\cdot{})$
function and an application of Jensen's inequality~\cite[Pg. 25]{CoverThomas}. We also use the fact that 
$\frac{1}{m}\sum_{i=1}^m\sigma_{G,i}^2= d_m(L)\sigma_G^2$, which can be proven as follows
\begin{eqnarray}
\nonumber\frac{1}{m}\expect{\sum_{i=1}^m Z_{L,i}^2 }&\overset{(\text{using}~\eqref{eq:fz})}{=}& \frac{\sigma_G^2}{m}\int_{\m{z}\in\mathcal{S}_L^G}\frac{\|\m{z}\|^2}{\sigma_G^2} c_m(L)\frac{\exp\left(-\frac{\|\m{z}_G\|^2}{2\sigma_G^2}\right)}{\left(\sqrt{2\pi\sigma_G^2}\right)^m}d\m{z}_G\\
\nonumber & =& \frac{c_m(L)\sigma_G^2}{m}\expect{\|\m{Z}_G\|^2\indi{\|\m{Z}_G\|\leq \sqrt{mL^2\sigma_G^2}}}\\
\nonumber&\overset{(\m{\widetilde{Z}}:=\frac{\m{Z}_G}{\sigma_G})}{=}&\frac{c_m(L)\sigma_G^2}{m}\expect{\|\m{\widetilde{Z}}\|^2\indi{\|\m{\widetilde{Z}}\|\leq \sqrt{mL^2}}}\\
\nonumber&=& \frac{c_m(L)\sigma_G^2}{m}\bigg(\expect{\|\m{\widetilde{Z}}\|^2}-\expect{\|\m{\widetilde{Z}}\|^2\indi{\|\m{\widetilde{Z}}\|> \sqrt{mL^2}}}\bigg)\\
\nonumber&\overset{(\text{using}~\eqref{eq:psinplus2})}{=}&\frac{c_m(L)\sigma_G^2}{m}\left(m-m\psi(m+2,\sqrt{mL^2})\right)\\
 &=&c_m(L)\left(1-\psi(m+2,L\sqrt{m})  \right)\sigma_G^2 =  d_m(L)\sigma_G^2.
\label{eq:expectzl}
\end{eqnarray}
We now compute $h(\m{Z}_L)$
\begin{eqnarray}
\label{eq:hz}
\nonumber  h(\m{Z}_L)&=&\int_{\m{z}\in \mathcal{S}_L^G}f_{Z_L}(\m{z})\lo{\frac{1}{f_{Z_L}(\m{z})}}d\m{z}=\int_{\m{z}\in \mathcal{S}_L^G}f_{Z_L}(\m{z})\lo{\frac{\left(\sqrt{2\pi \sigma_G^2}\right)^m}{c_m(L)e^{-\frac{\|\m{z}\|^2}{2\sigma_G^2}}}}d\m{z}\\
&=& -\lo{c_m(L)}+\frac{m}{2}\lo{2\pi\sigma_G^2}+\int_{\m{z}\in\mathcal{S}_L^G}c_m(L)f_{G}(\m{z})\frac{\|\m{z}\|^2}{2\sigma_G^2}\lo{e}d\m{z}.
\end{eqnarray}
Analyzing the last term of~\eqref{eq:hz},
\begin{eqnarray}
\nonumber \int_{\m{z}\in\mathcal{S}_L^G}c_m(L)f_{G}(\m{z})\frac{\|\m{z}\|^2}{2\sigma_G^2}\lo{e}d\m{z} &=&\frac{\lo{e}}{2\sigma_G^2} \int_{\m{z}\in\mathcal{S}_L^G}c_m(L)\frac{     e^{-\frac{\|\m{z}\|^2}{2\sigma_G^2}}     } {  \left(\sqrt{2\pi \sigma_G^2}\right)^m  }\|\m{z}\|^2d\m{z}\\
\nonumber &=&\frac{\lo{e}}{2\sigma_G^2} \int_{\m{z}}f_{Z_L}(\m{z})\|\m{z}\|^2d\m{z}\\
\nonumber &\overset{(\text{using}~\eqref{eq:fz})}{=}&\frac{\lo{e}}{2\sigma_G^2}\expectp{G}{\|\m{Z}_L\|^2}  = \frac{\lo{e}}{2\sigma_G^2}  \expectp{G}{\sum_{i=1}^m  Z_{L,i}^2 }  \\
&\overset{(\text{using}~\eqref{eq:expectzl})}{=}& \frac{\lo{e}}{2\sigma_G^2}md_m(L)\sigma_G^2=\frac{m\lo{e^{d_m(L)}}}{2}. 
\label{eq:lastterm}
\end{eqnarray}
The expression $C_{G,L}^{(m)}$ can now be upper bounded using~\eqref{eq:cn},~\eqref{eq:hz} and~\eqref{eq:lastterm} as follows.
\begin{eqnarray}
\nonumber C_{G,L}^{(m)}&\leq& \frac{1}{2}\lo{2\pi e (\bar{P}+d_m(L)\sigma_G^2)}+\frac{1}{m}\lo{c_m(L)} - \frac{1}{2}\lo{2\pi\sigma_G^2}-\frac{1}{2}\lo{e^{d_m(L)}}\\
\nonumber &=& \frac{1}{2}\lo{2\pi e (\bar{P}+d_m(L)\sigma_G^2)}+\frac{1}{2}\lo{c_m^{\frac{2}{m}}(L)} - \frac{1}{2}\lo{2\pi\sigma_G^2}-\frac{1}{2}\lo{e^{d_m(L)}}\\
 &= & \frac{1}{2}\lo{\frac{2\pi e  (\bar{P}+d_m(L)\sigma_G^2)c_m^{\frac{2}{m}}(L)}{2\pi \sigma_G^2 e^{d_m(L)}} }=  \frac{1}{2}\lo{\frac{e^{1-d_m(L)} (\bar{P}+d_m(L)\sigma_G^2)c_m^{\frac{2}{m}}(L)}{\sigma_G^2} }. 
\label{eq:capacitybd}
\end{eqnarray}
\end{proof}
Now, recall that the rate-distortion function $D_m(R)$ for squared error distortion for source $\m{X}_0$ and reconstruction $\whatmn{X}_L$ is,
\begin{equation}
D_m(R):=
\inf_{\scriptsize \begin{array}{c}
p(\whatmn{X}_L|\m{X}_0)\\
\frac{1}{m}I(\m{X}_0;\whatmn{X}_L)\leq R
\end{array}}
\frac{1}{m}\expectp{\m{X}_0,\m{Z}_G}{\|\m{X}_0-\whatmn{X}_L\|^2},
\end{equation}
which is the dual of the rate-distortion function~\cite[Pg. 341]{CoverThomas}. 
Since $I(\m{X}_0;\whatmn{X}_L)\leq mC_{G,L}^{(m)}$, using the converse to
the rate distortion theorem~\cite[Pg. 349]{CoverThomas} and the upper
bound on the mutual information represented by $C_{G,L}^{(m)}$, 
\begin{equation}
\label{eq:ratedist}
\frac{1}{m} \expectp{\m{X}_0,\m{Z}_G}{\|\m{X}_0-\whatmn{X}_L\|^2} \geq D_m(C_{G,L}^{(m)}).
\end{equation}
Since the Gaussian source is iid, $D_m(R)=D(R)$, where
$D(R)=\sigma_0^22^{-2R}$ is the distortion-rate function for a
Gaussian source of variance
$\sigma_0^2$~\cite[Pg. 346]{CoverThomas}. Thus,
using~\eqref{eq:sqrtd},~\eqref{eq:xl} and~\eqref{eq:ratedist},   
\begin{eqnarray*}
\expectp{\m{X}_0,\m{Z}_G}{J_2^{(\gamma)}(\m{X}_0,\m{Z})|\m{Z}\in\mathcal{S}_L^G} \geq \left(\left(\sqrt{D(C_{G,L}^{(m)})} - \sqrt{P} \right)^+\right)^2.
\end{eqnarray*}
Substituting the bound on $C_{G,L}^{(m)}$ from~\eqref{eq:capacitybd},
\begin{eqnarray*}
D(C_{G,L}^{(m)})= \sigma_0^2  2^{-2C_{G,L}^{(m)}} =\frac{\sigma_0^2\sigma_G^2}{c_m^{\frac{2}{m}}(L)e^{1-d_m(L)} (\bar{P}+d_m(L)\sigma_G^2)}.
\end{eqnarray*}
Using~\eqref{eq:sqrtd}, this completes the proof of the lemma. Notice
that $c_m(L)\rightarrow 1$ and $d_m(L)\rightarrow 1$ for fixed $m$ as $L\rightarrow\infty$, as well as for fixed $L>1$ as $m\rightarrow\infty$. So the lower bound on $D(C_{G,L}^{(m)})$ approaches $\kappa$ of Theorem~\ref{thm:oldbound} in both of 
these two limits. 


%

\bibliographystyle{IEEEtran}
\bibliography{IEEEabrv,MyMainBibliography}
\end{document}